\documentclass[a4paper,fleqn]{article}
\usepackage{a4wide}
\usepackage{amsmath,amssymb,amsthm}
\usepackage{thmtools}
\usepackage{enumerate}
\usepackage{xifthen}
\usepackage{graphicx}
\usepackage{tikz}
\usepackage{tikz-qtree} % Strategiapuita varten.
\tikzset{grow'=right} % Make trees grow from left to right.
\tikzset{every tree node/.style={anchor=base west}} % Align nodes of the tree to the left (west).
\usepackage{abstract}
\usepackage{authblk}

\usepackage{hyperref}
\usepackage[T1]{fontenc}
\usepackage{mathpazo}
\makeatletter\@ifpackageloaded{mathpazo}\@tempswatrue\@tempswafalse
\if@tempswa
  \DeclareFontFamily{OT1}{pzc}{}
  \DeclareFontShape{OT1}{pzc}{m}{it}{<-> s * [1.15] pzcmi7t}{}
  \DeclareMathAlphabet{\mathpzc}{OT1}{pzc}{m}{it}
\fi\makeatother
\usepackage{bm}
\usepackage[utf8]{inputenc}

\usepackage[sorting=nyt,giveninits=true,maxbibnames=4,backend=biber]{biblatex}
\makeatletter\@ifpackageloaded{biblatex}{%
  \usepackage{csquotes} % Silence a warning.
  \bibliography{../../references}
  \renewbibmacro{in:}{%
    \ifentrytype{incollection}{\printtext{\bibstring{in}\intitlepunct}}{}}
  \renewbibmacro{publisher+location+date}{%
    \iflistundef{publisher}
      {\setunit*{\addcomma\space}}
      {\setunit*{\addcomma\space}}%
    \printlist{publisher}%
    \setunit*{\addcomma\space}%
    \printlist{location}%
    \setunit*{\addcomma\space}%
    \usebibmacro{date}%
    \newunit}
  \DeclareFieldFormat[article]{pages}{#1\isdot}
  \DeclareFieldFormat[article,incollection,inproceedings,unpublished]{title}{#1\isdot}
  \DeclareFieldFormat[thesis]{title}{\mkbibemph{#1\isdot}}
  \DeclareFieldFormat[unpublished]{date}{(#1)\isdot}
  \DeclareFieldFormat[unpublished]{note}{#1\nopunct} % removes the unnecessary dot before the year
  \DeclareFieldFormat[article]{journaltitle}{\mkbibemph{#1\isdot}}
  
  \AtEveryBibitem{%
    \ifentrytype{book}{}{% Remove editor except for books
      \clearname{editor}
    }
  }
  \newbibmacro*{bbx:parunit}{%
    \ifbibliography
      {\setunit{\bibpagerefpunct}\newblock
       \usebibmacro{pageref}%
       \clearlist{pageref}%
       \setunit{\adddot\par\nobreak}}
      {}
  }
  \renewbibmacro*{doi+eprint+url}{%
    \usebibmacro{bbx:parunit}% Added
    \iftoggle{bbx:doi}
      {\printfield{doi}}
      {}%
    \iftoggle{bbx:eprint}
      {\usebibmacro{eprint}}
      {}%
    \iftoggle{bbx:url}
      {\usebibmacro{url+urldate}}
      {}
  }
  \renewbibmacro*{eprint}{%
    \usebibmacro{bbx:parunit}% Added
    \iffieldundef{eprinttype}
      {\printfield{eprint}}
      {\printfield[eprint:\strfield{eprinttype}]{eprint}}
  }
  \renewbibmacro*{url+urldate}{%
    \usebibmacro{bbx:parunit}% Added
    \printfield{url}%
    \iffieldundef{urlyear}
      {}
      {\setunit*{\addspace}%
       \printtext[urldate]{\printurldate}}
  }
}{}\makeatother

\declaretheorem[numberwithin=section,refname={theorem,theorems},Refname={Theorem,Theorems}]{theorem}
\declaretheorem[sibling=theorem,style=definition]{definition}
\declaretheorem[sibling=theorem,style=definition,name=Example]{example}
\declaretheorem[sibling=theorem,style=definition,name=Remark]{remark}
\declaretheorem[sibling=theorem,name=Lemma]{lemma}
\declaretheorem[sibling=theorem,name=Proposition]{proposition}
\declaretheorem[sibling=theorem,name=Corollary]{corollary}

\declaretheorem[sibling=theorem,name=Question]{question}
\declaretheorem[sibling=theorem,name=Problem]{problem}

\makeatletter\@ifpackageloaded{hyperref}{%
  \usepackage{xcolor}
  \definecolor{dark-red}{rgb}{0.4,0.15,0.15}
  \definecolor{dark-blue}{rgb}{0.15,0.15,0.4}
  \definecolor{medium-blue}{rgb}{0,0,0.5}
  \hypersetup{
    colorlinks,
    linkcolor={dark-red},
    citecolor={dark-blue},
    urlcolor={medium-blue}%
  }
}{}\makeatother

\usepackage{sectsty}
\sectionfont{\large\bfseries}
\subsectionfont{\normalsize\bfseries}

\newcommand{\infw}[1]{%
  \ifcat\noexpand#1\relax\bm{#1}% check if the argument is a control sequence, i.e., probably greek letter
  \else\mathbf{#1}\fi}          % should be a latin letter
\newcommand{\OC}[1]{\overline{\mathcal{O}(#1)}}
\newcommand{\ans}[1]{\mathcal{#1}}
\newcommand{\rep}[1][]{\mathrm{rep}_{#1}}
\newcommand{\val}[1][]{\mathrm{val}_{#1}}
\newcommand{\supp}{\mathrm{supp}}
\newcommand{\Lang}[2][]{\mathcal{L}_{#2}\ifthenelse{\isempty{#1}}{}{(#1)}}

\newcommand{\N}{\mathbb{N}}

\newcommand{\keywords}[1]{\par\noindent{\footnotesize{\em Keywords\/}: #1}}

\title{Automatic winning shifts}
\author[,1,2,3]{Jarkko Peltomäki\footnote{Corresponding author.\\E-mail addresses:
\href{mailto:r@turambar.org}{r@turambar.org} (J. Peltomäki), \href{mailto:vosalo@utu.fi}{vosalo@utu.fi} (V. Salo).}}
\affil[1]{The Turku Collegium for Science, Medicine and Technology TCSMT, University of Turku, Turku, Finland}
\affil[2]{Turku Centre for Computer Science TUCS, Turku, Finland}
\affil[3]{University of Turku, Department of Mathematics and Statistics, Turku, Finland}
\author[3]{Ville Salo}
\date{}

\begin{document}
\maketitle
\vspace{1em}
\noindent
\hrulefill
\begin{abstract}
  \vspace{-0.5em}
  \noindent
  To each one-dimensional subshift $X$, we may associate a winning shift $W(X)$ which arises from a combinatorial game
  played on the language of $X$. Previously it has been studied what properties of $X$ does $W(X)$ inherit. For
  example, $X$ and $W(X)$ have the same factor complexity and if $X$ is a sofic subshift, then $W(X)$ is also sofic. In
  this paper, we develop a notion of automaticity for $W(X)$, that is, we propose what it means that a vector
  representation of $W(X)$ is accepted by a finite automaton.

  Let $S$ be an abstract numeration system such that addition with respect to $S$ is a rational relation. Let $X$ be a
  subshift generated by an $S$-automatic word. We prove that as long as there is a bound on the number of nonzero
  symbols in configurations of $W(X)$ (which follows from $X$ having sublinear factor complexity), then $W(X)$ is
  accepted by a finite automaton, which can be effectively constructed from the description of $X$. We provide an
  explicit automaton when $X$ is generated by certain automatic words such as the Thue-Morse word.
  \vspace{1em}
  \keywords{winning shift, combinatorial game, abstract numeration system, automatic sequence, regular language, sofic shift}
  \vspace{-1em}
\end{abstract}
\hrulefill

\section{Introduction}
Consider a \emph{target set} $X$ of words of a common length $n$ written over an alphabet $A$. Let
$\alpha_1 \dotsm \alpha_n$ be a \emph{choice sequence} of integers in $\{0, 1, \ldots, |A| - 1\}$. Given the choice
sequence, Alice and Bob play a game as follows. On round $j$, $1 \leq j \leq n$, Alice chooses a subset $A_j$ of $A$ of
size $\alpha_j + 1$. Then Bob picks a letter $a_j$ from $A_j$. After $n$ rounds, a word $a_1 \dotsm a_n$ is built. If
this word is in $X$, then Alice wins, and Bob wins otherwise.

Let $W(X)$ be the set of choice sequences for which Alice has a winning strategy. This set is called the \emph{winning
set} of $X$. For example, if $X = \{000, 110, 111\}$ and the choice sequence is $101$, then Bob has a winning strategy.
Indeed on the first turn Alice must allow Bob to choose from $A_1 = \{0, 1\}$ (since this is the only subset of size
$2$ of the alphabet), and Bob may pick the letter $0$. Alice must then select the subset $A_1 = \{0\}$ since no word in
$X$ begins with $01$. On the final round Bob again has two choices, and he may select $1$. This results in the word
$001$ which is not in $X$. Hence $101 \notin W(X)$. It is straightforward to verify that $W(X) = \{000, 001, 100\}$.

More generally, the game can be allowed to have infinitely many turns, and the winning set makes sense for sets $X$
containing words of different lengths by setting
\begin{equation*}
  W(X) = \bigcup_{n \in \N \cup \{\N\}} W(X \cap A^n).
\end{equation*}
The winning set was introduced in the paper \cite{2014:playing_with_subshifts} by the second author and I. Törmä. The
paper \cite{2014:playing_with_subshifts} contains rigorous definitions of the preceding concepts and proofs of the
basic properties of winning sets. A key property of $W(X)$ is that it is \emph{hereditary} (or downward-closed): if $u$
and $v$ are words of equal length satisfying $u \leq v$ and $v \in W(X)$, then $u \in W(X)$ (here $\leq$ is the
coordinatewise order induced by the natural order on $\N$). Most importantly, if $X$ is a subshift, then $W(X)$, now
called the \emph{winning shift} of $X$, is also a subshift. As a terminological point, in
\cite{2014:playing_with_subshifts} the term ``winning shift'' and the notation $W(X)$ refers to a slightly different
subshift in the case of a nonbinary alphabet, and our $W(X)$ is denoted by $\tilde W(X)$.

Several properties of $X$ are inherited by $W(X)$. For example, if $X$ is a regular language, then $W(X)$ is also
regular. Surprisingly factor complexity is preserved: there are equally many words of length $n$ in $X$ and $W(X)$ for
all $n$. In a sense, the mapping $X \mapsto W(X)$ is a factor complexity preserving map that reorganizes $X$ to a
hereditary set. This property was used in \cite{2019:on_winning_shifts_of_marked_uniform_substitutions} to rederive
results of \cite{1998:on_uniform_d0l_words} that allow to determine the factor complexity of fixed points of marked
uniform substitutions. Additionally, it was shown that $W(X)$ has a substitutive structure when $X$ is a subshift
generated by a marked uniform substitution. In
\cite{2020:descriptional_complexity_of_winning_sets_of_regular_languages}, Marcus and Törmä study the descriptional
complexity of $W(X)$ when $X$ is regular. A concept equivalent to winning sets has been introduced in the context of
set systems in \cite{2002:shattering_news}. In \cite{2021:trees_in_positive_entropy_subshifts}, winning shifts are
utilized to show that positive entropy implies that a subshift contains a ``steadily branching binary tree'' (which in
the case of a binary alphabet implies positive independence entropy).

The aim of this paper is to study the structure of the winning shift $W(X)$ of a subshift $X$ generated by an automatic
word. Automatic words are a well-studied class of words which are generated by finite state machines. In fact, the
above-mentioned fixed points of marked uniform substitutions are automatic words.

Our setup is as follows. We consider representations of integers in an abstract numeration system $\ans{S}$ and define
$\ans{S}$-automatic words as words obtained by feeding representations of integers to finite automata with output. We
then develop the notion of an $\ans{S}$-codable set of infinite words over $\N$ as a set of words whose supports,
represented in $\ans{S}$, can be recognized by a finite automaton. Our main result states that if the addition relation
of $\ans{S}$ is rational and $X$ is a subshift generated by an $\ans{S}$-automatic word with sublinear factor
complexity, then its winning shift $W(X)$ is $\ans{S}$-codable. We obtain as a corollary that if $\ans{S}$ is a Pisot
numeration system and $X$ is a subshift generated by an $\ans{S}$-automatic word, then $W(X)$ is $\ans{S}$-codable.
Pisot-automatic words include the well-known $k$-automatic words as well as Fibonacci-automatic words (such as the
Fibonacci word).

We also discuss the necessity of the assumptions of our main result. For the result, it is crucial that $W(X)$ is
countable. We give an example due to J.\ Cassaigne of a subshift $X$ generated by an $\ans{S}$-automatic word with
superlinear factor complexity for which $W(X)$ is uncountable, and we exhibit a subshift $X$ with sublinear factor
complexity such that $W(X)$ is uncountable.

The proof of the main result is constructive, so in principle an automaton for the winning shift can be found
algorithmically. We determine the automaton explicitly for the winning shifts of the subshifts generated by the
following $2$-automatic words: the Thue-Morse word, the regular paperfolding word, the Rudin-Shapiro word, and the
period-doubling word.

This paper also includes some general results on the $\ans{S}$-codability of winning shifts of sofic shifts, as well as
basic robustness properties of the class of $\ans{S}$-codable subshifts, in particular closure under topological
conjugacy.

\section{Preliminaries}

We assume that the reader is familiar with basic formal language theory. We point the reader to
\cite{1997:handbook_of_formal_languages_vol1} for information on automata and regular languages and
\cite{2002:algebraic_combinatorics_on_words} on combinatorics on words. As a terminological point, ``regular'' and
``rational'' both refer to regular languages, but we use the word ``rational'' to emphasize that a language encodes a
relation.

In this paper, $\N = \{0, 1, \ldots\}$, and we index words from $0$ unless stated otherwise. Let $A$ be an alphabet
and $A^\N$ be the set of right-infinite words over $A$. When $\infw{x} = x_0 x_1 \dotsm$ with $x_i \in A$, then with
$\infw{x}[i,j]$, $i \leq j$, we refer to the factor $x_i \dotsm x_j$ of $\infw{x}$. We use the shorthand $\infw{x}[i]$
for $\infw{x}[i,i]$. The set of factors, the \emph{language} of $\infw{x}$ is denoted by $\Lang[\infw{x}]{}$. The
\emph{factor complexity function} $\rho_{\infw{x}}$ counts the number of distinct factors of length $n$ in $\infw{x}$,
that is, $\rho_{\infw{x}}(n) = |\{\infw{x}[p,p+n-1] : p \geq 0\}|$. If $\rho_{\infw{x}}(n) = \mathcal{O}(n)$, then we
say that $\infw{x}$ has \emph{sublinear factor complexity}. Otherwise we say that $\infw{x}$ has \emph{superlinear
factor complexity}.

A \emph{subshift} $X$ is a shift-invariant subset of $A^\N$ which is compact in the product topology on $A^\N$ with the
discrete topology on a finite set $A$. We also consider subshifts in $\N^\N$, and they are defined analogously by
letting $\N$ to have the discrete topology. Compactness implies that the alphabet must actually be finite, that is, up
to renaming of letters, we have $X \subseteq A^\N$ for a finite $A$. Indeed, otherwise the open cover
$\{\{x_0 x_1 \dotsm \in X : x_0 = n\} : n \in \N\}$ of $X$ has no finite subcover. A subshift $X$ is determined by its
\emph{language} $\Lang[X]{}$ which is the set of factors of all words in $X$. By $\OC{\infw{x}}$ we mean the closure of
the orbit of $\infw{x}$ under the shift map. If $X = \OC{\infw{x}}$, then we say that $X$ is \emph{generated} by
$\infw{x}$. A subshift $X$ is \emph{transitive} if it is generated by $\infw{x}$ for some $\infw{x} \in X$.

Suppose that $A$ has a total order $<$. A subset $X$ of $A^n$ with $n \in \N \cup \{\N\}$ is \emph{hereditary} if
$v \in X$ and $u \leq v$ coordinatewise imply $u \in X$ for all $v \in X$ and all $u \in A^n$. If $X$ is a subshift,
then the smallest hereditary subshift containing $X$ is called the \emph{hereditary closure} of $X$. In the earlier
papers
\cite{2014:playing_with_subshifts,2019:on_winning_shifts_of_marked_uniform_substitutions,2020:descriptional_complexity_of_winning_sets_of_regular_languages},
the word downward-closed was used, but we opt to use hereditary which seems more common in symbolic dynamics
\cite{2007:independence_in_topological_and_c-dynamics,2013:topological_entropy_and_distributional_chaos_in_hereditary}.

Each word in the winning set $W(X)$ of $X$ corresponds to at least one winning strategy for Alice. It is useful to
think of the winning strategies as strategy trees.  When
\begin{equation*}
  X = \{0010, 0011, 0100, 0101, 0110, 1001, 1010, 1011, 1100, 1101\},
\end{equation*}
then
\begin{equation*}
  W(X) = \{0000, 0001, 0010, 0100, 0101, 1000, 1001, 1010, 1100, 1101\}.
\end{equation*}
Alice's winning strategy for the choice sequence $1101$ is depicted as the strategy tree of
\autoref{fig:thue-morse_strategy}. The interpretation is that on the first turn Alice has to offer two choices for Bob,
so the tree branches. The same branching happens on the second round. The letters immediately following the second
branching dictate which letter Alice should force Bob to choose on the third round. As the choice sequence ends with
$1$, the tree then branches once more. The paths from the root to the leaves form all possible words played when Alice
uses this strategy, and indeed they all belong to $X$.

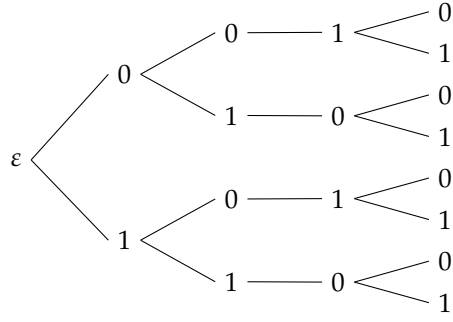
\begin{figure}
  \centering
  \centering
  \begin{tikzpicture}
    \tikzset{level distance=40pt}
    \Tree [.$\varepsilon$ [.$0$ [.$0$ [.$1$ [.$0$ ] [.$1$ ] ] ] [.$1$ [.$0$ [.$0$ ] [.$1$ ] ] ] ] [.$1$ [.$0$ [.$1$ [.$0$ ] [.$1$ ] ]
    ] [.$1$ [.$0$ [.$0$ ] [.$1$ ] ] ] ] ]
  \end{tikzpicture}
  \caption{Alice's winning strategy depicted as a strategy tree.}\label{fig:thue-morse_strategy}
\end{figure}

\section{Numeration Systems and \texorpdfstring{$\ans{S}$}{S}-automatic Words}

\subsection{Numeration Systems}
Let $L$ be an infinite language with a fixed total order $\prec$ of order type $\omega$, and let
$\ans{S} = (L, \prec)$. We call the pair $\ans{S}$ an \emph{abstract numeration system} which we abbreviate as
\emph{ANS}. Usually it is required in an ANS $\ans{S}$ that $L$ is regular and $\prec$ is the radix order induced by a
total order on the alphabet meaning short words come before long ones and words of the same length are ordered
lexicographically. Our definition is slightly more general since our main results work for any order $\prec$ under the
orthogonal (but morally stronger) assumption of ``addability''; see below for the definition. See
\cite{2001:numeration_systems_on_a_regular_language}, \cite[Ch.~3]{2010:combinatorics_automata_and_number_theory}, and
the recent paper \cite{2022:regular_sequences_and_synchronized_sequences_in_abstract} and its references for more on
ANS. The order $\prec$ naturally yields a bijection $\theta\colon L \to \N$. For $w \in L$ and $n \in \N$, we set
$\val[\ans{S}](w) = \theta(w)$ and $\rep[\ans{S}](n) = \theta^{-1}(n)$. We omit the subscripts $\ans{S}$ whenever they
are clear from context.

A common example of an ANS is a positional numeration system. Let $(U_i)$ be a strictly increasing sequence of positive
integers with $U_0 = 1$. Represent an integer $n$ greedily as $n = d_k U_k + \dotsm + d_0 U_0$ and set
$\rep(n) = d_k \dotsm d_0$. Then $(\rep(\N), \prec)$ is an ANS where $\prec$ is induced by the natural order on
$\N$. The usual base-$k$ representation of an integer is obtained by letting $U_i = k^i$ for all $i$. More generally if
$(U_i)$ satisfies a linear recurrence such that the characteristic polynomial of the recurrence is the minimal
polynomial of a Pisot number (a number whose conjugates $\alpha$ satisfy $|\alpha| < 1$), then we obtain a Pisot
numeration system. If $(U_i)$ is the sequence of Fibonacci numbers, then we obtain the Fibonacci numeration system
which is an example of a Pisot numeration system. Pisot numeration systems enjoy many good properties; see, e.g.,
\cite[Ch.~2.4]{2014:formal_languages_automata_and_numeration_systems_2} and
\cite[Ch.~2]{2010:combinatorics_automata_and_number_theory} and their references.

Next we extend the functions $\val[\ans{S}]$ and $\rep[\ans{S}]$ to tuples of words and integers. Let $A$ be an
alphabet and
\begin{equation*}
  \begin{bmatrix}
    w_1 \\
    \vdots \\
    w_d
  \end{bmatrix}
\end{equation*}
be an element of $A^* \times \dotsm \times A^*$ for a positive integer $d$. The words $w_1$, $\ldots$, $w_d$ might be
of distinct lengths and we cannot feed them to an automaton in parallel, so we pad them to equal length as follows. If
$\#$ is a letter that does not belong to $A$, then we set
\begin{equation*}
  \begin{bmatrix}
   w_1 \\
   \vdots \\
   w_d
  \end{bmatrix}^{\#}
  =
  \begin{bmatrix}
    \#^{M - |w_1|} w_1 \\
    \vdots \\
    \#^{M - |w_d|} w_d
  \end{bmatrix}
\end{equation*}
where $M = \max\{|w_1|, \ldots, |w_d|\}$. Thus by padding, we may work on tuples of words as if they are words over the
alphabet $(A \cup \{\#\})^d$.

Let $\ans{S}$ be an ANS, $\#$ be a letter not appearing in the language $L$ of $\ans{S}$, $d$ be a positive integer,
and $\mathbf{L} = (L^d)^{\#}$. We extend $\rep[\ans{S}]$ and $\val[\ans{S}]$ to elements of $\N^d$ and $\mathbf{L}$ as
follows:
\begin{equation*}
  \rep[\ans{S}]\colon \N^d \to \mathbf{L},
  \begin{bmatrix}
    n_1 \\
    \vdots \\
    n_d
  \end{bmatrix}
  \mapsto
  \begin{bmatrix}
    \rep[\ans{S}](n_1) \\
    \vdots \\
    \rep[\ans{S}](n_d)
  \end{bmatrix}^{\#}
\end{equation*}
and
\begin{equation*}
  \val[\ans{S}]\colon \mathbf{L} \to \N^d,
  \begin{bmatrix}
    w_1 \\
    \vdots \\
    w_d
  \end{bmatrix}
  \mapsto
  \begin{bmatrix}
    \val[\ans{S}](\tau_\#(w_1)) \\
    \vdots \\
    \val[\ans{S}](\tau_\#(w_d))
  \end{bmatrix}
\end{equation*}
where $\tau_{\#}$ is a substitution erasing the letter $\#$. Notice that if $L$ is regular, then $(L^d)^\#$ is also
regular.

The following definition is central to this paper. It defines what it means that a set of integers or integer tuples
is accepted by a finite automaton.

\begin{definition}
  Let $\ans{S}$ be an ANS. A subset $Y$ of $\N^d$ is \emph{$\ans{S}$-recognizable} if $\rep[\ans{S}](Y)$ is regular. If
  $Y \subseteq \N^{\leq d} = \N \cup \N^2 \cup \dotsm \cup \N^d$, we say that $Y$ is $\ans{S}$-recognizable if each
  $Y \cap \N^i$, $i = 1, \ldots, d$, is $\ans{S}$-recognizable.
\end{definition}

In the usual base-$k$ numeration system $\ans{B}$, the addition relation is rational meaning that the set
$\{(x, y, x+y) : x, y \in \N\}$ is $\ans{B}$-recognizable. This also holds for Pisot numeration systems
\cite{1992:representations_of_numbers_and_finite_automata}, but it is not a property of a general ANS $\ans{S}$.
Indeed, even a constant multiple $tY$ of an $\ans{S}$-recognizable set $Y$ is not always $\ans{S}$-recognizable; see
\cite[Section~3.3.3]{2010:combinatorics_automata_and_number_theory}. We give the following definitions.

\begin{definition}
  An ANS $\ans{S}$ is \emph{regular} if the set $\N$ is $\ans{S}$-recognizable. 
  An ANS $\ans{S}$ is \emph{comparable} if the set $\{(x, y) : x \leq y \}$ is $\ans{S}$-recognizable. 
  An ANS $\ans{S}$ is \emph{addable} if the set $\{(x, y, x+y) : x, y \in \N\}$ is $\ans{S}$-recognizable. 
\end{definition}

\begin{lemma}\label{lem:addable_implications}
  If an ANS is addable, then it is comparable. If an ANS is comparable, then it is regular.
\end{lemma}
\begin{proof}
  Up to (easy) padding issues, comparability follows from addability by projecting to the first and third component,
  and regularity follows from comparability by projecting to the first component.
\end{proof}

The literature contains many examples of ANS that are not Pisot and not addable, but we are unaware of any ANS using
the radix order which is not Pisot and is addable.

\subsection{\texorpdfstring{$\ans{S}$}{S}-automatic Words}
We only give the definition of $\ans{S}$-automatic words. For a more comprehensive introduction to the subject, we
refer the reader to \cite{2003:automatic_sequences}.

A \emph{deterministic finite automaton with output}, or DFAO, is a finite automaton with an output function $\tau$
associated with the set of states. When a word $w$ is fed to the automaton and state $q$ is reached, the output of the
automaton with input $w$ is defined as $\tau(q)$.

\begin{definition}
  Let $\ans{S}$ be an ANS. An infinite word $\infw{x}$ is \emph{$\ans{S}$-automatic} if there exists a DFAO
  $\mathcal{A}$ such that $\infw{x}[n]$ equals the output of $\mathcal{A}$ with input $\rep[S](n)$ for all $n \geq 0$.
\end{definition}

\section{Weakly \texorpdfstring{$\ans{S}$}{S}-codable and \texorpdfstring{$\ans{S}$}{S}-codable Sets}\label{sec:s_codable}
In this section, we show how to represent sequences in $\N^\N$ so that their supports can be recognized by a finite
automaton.

Consider an infinite word $\infw{x}$ in $\N^\N$. If $\infw{x} = x_0 x_1 \dotsm$, $x_i \in \N$, then we define
\begin{equation*}
  \sum \infw{x} = \sum_{i \in \N} x_i.
\end{equation*}
The \emph{support} $\supp(\infw{x})$ of $\infw{x}$ is the set $\{n \in \N : x_n \neq 0\}$.

\begin{definition}
  Let $\infw{x}$ in $\N^\N$ be such that $\sum \infw{x} = d < \infty$. Let $\nu(\infw{x})$ be the unique vector
  $(n_1, \ldots, n_d)$ in $\N^d$ such that $n_i \leq n_{i+1}$ for all $i$ and $x_j = |\{k : n_k = j\}|$ for all
  $j \in \N$.
\end{definition}

For example, if $\infw{x} = 1002010^\omega$, then $\nu(\infw{x}) = (0, 3, 3, 5)$. In other words, we repeat the indices
in the support as many times as indicated by the letters: at index $0$ we have $1$, so the index $0$ is repeated once;
at index $3$ we have $2$, so the index $3$ is repeated twice; etc. Clearly $\infw{x}$ can be uniquely reconstructed
from $\nu(\infw{x})$: $\infw{x}$ is the sum of the characteristic functions of singleton sets $\{i\}$ where $i$ takes
on the values in $\nu(\infw{x})$ (with repetitions).

\begin{definition}
  Let $Y$ be a subset of $\N^\N$. Its \emph{coding dimension} is the smallest integer $d \in \N$
  such that $\sum \infw{y} \leq d$ for all $\infw{y} \in Y$, if such a $d$ exists.
\end{definition}

\begin{definition}
  Let $\ans{S}$ be an ANS. A subset $Y$ of $\N^\N$ is \emph{weakly $\ans{S}$-codable} if for all nonnegative integers
  $k$ the set
  \begin{equation*}
    \left\{\nu(\infw{y}) : \infw{y} \in Y, \sum \infw{y} \leq k \right\}
  \end{equation*}
  is $\ans{S}$-recognizable. The subset $Y$ is \emph{$\ans{S}$-codable} if it is weakly $\ans{S}$-codable and has
  finite coding dimension.
\end{definition}

Notice that if $Y$ is $\ans{S}$-codable, then $\nu(Y)$ is $\ans{S}$-recognizable since the union of finitely many
regular languages is regular. Notice also that weak $\ans{S}$-codability considers only words with finite support. Thus
if $Y$ consists of infinite words over $\{0, 1\}$ having infinitely many occurrences of $1$, then $Y$ is trivially
weakly $\ans{S}$-codable for any $\ans{S}$ since the sets in the above definition are all empty or contain an empty
vector. The notion makes most sense for subshifts where configurations with finite sum are dense. For example all
hereditary subshifts have this property.

$\ans{S}$-codability and weak $\ans{S}$-codability are relatively robust notions, see \autoref{sec:Robust}.

\section{Weakly \texorpdfstring{$\ans{S}$}{S}-codable Winning Shifts}
In this section, we introduce the necessary results and prove that the winning shift of an $\ans{S}$-automatic word is
weakly $\ans{S}$-codable when $\ans{S}$ is addable; see \autoref{thm:ws_weakly_codable}.

First we introduce another way to represent words in $\N^\N$. This representation allows more convenient proofs and
leads to the same $\ans{S}$-codability properties as in \autoref{sec:s_codable}.

Let $\infw{x} \in \N^\N$, and say $\infw{x}$ has finite support. View the support $\supp(\infw{x})$ as a vector
$(n_1, \ldots, n_k)$ with $n_i < n_{i+1}$ for all $i$, and denote this vector by $s(\infw{x})$. Let moreover
$e(\infw{x})$ be the word obtained from $\infw{x}$ by erasing all letters $0$ (this notion naturally extends to words
with infinite supports).

\begin{definition}
  Let $v$ be a word over $\N_{>0}$ and $Y$ be a subset of $\N^\N$. Define
  \begin{equation*}
    Q_v(Y) = \{\infw{y} \in Y : e(\infw{y}) = v\} \quad \text{and} \quad P_v(Y) = s(Q_v(Y)).
  \end{equation*}
\end{definition}

For example, the set $P_{111}(Y)$ is the set of $3$-tuples $(n_1, n_2, n_3)$ such that $n_1 < n_2 < n_3$ and there
exists $\infw{y}$ in $Y$ with letters $1$ at positions $n_1$, $n_2$, and $n_3$ while other positions of $\infw{y}$
equal $0$.

\begin{lemma}\label{lem:equivalence}
  Let $\ans{S}$ be an ANS and $Y$ be a subset of $\N^\N$. Then $Y$ is weakly $\ans{S}$-codable if and only if $P_v(Y)$
  is $\ans{S}$-recognizable for all words $v \in \N_{>0}^*$.
\end{lemma}
\begin{proof}
  Suppose that $Y$ is weakly $\ans{S}$-codable. Let $v \in \N_{>0}^*$, and suppose that the sum of letters of $v$
  equals $k$. Let $Y_k = \left\{\nu(\infw{y}) : \infw{y} \in Y, \sum \infw{y} = k \right\}$. By assumption, there
  exists an automaton $\mathcal{A}$ accepting $\rep(Y_k)$. Based on the letters of $v$, we modify $\mathcal{A}$ as
  follows. If the first letter of $v$ is $i$, then we accept only those words whose $i$ first components are equal. If
  the second letter of $v$ equals $j$, then we accept only the words whose components $i+1$, $\ldots$, $i+j$ are equal
  and greater than $i$. We repeat this for each letter and obtain an automaton $\mathcal{A}'$. Then we project the
  words accepted by $\mathcal{A}'$ in a suitable way: we project the first $i$ components to one component, the next
  $j$ components to one component, and so on. The resulting language is regular and it equals $\rep(P_v(Y))$.

  When $\rep(P_v(Y))$ is regular for a given $v \in \N_{>0}^*$, we can reverse the projection made in the previous
  paragraph to obtain a regular language $\mathcal{L}_v$. We have
  \begin{equation*}
    \rep(Y_k) = \bigcup_{\substack{v \in \N_{>0}^* \\ \sum v = k}} \mathcal{L}_v,
  \end{equation*}
  so $\rep(Y_k)$ is regular as the union is finite. It follows that $Y$ is weakly $\ans{S}$-codable.
\end{proof}

Notice that it follows from \autoref{lem:equivalence} that $Y$ is $\ans{S}$-codable if and only if $P_v(Y)$ is
$\ans{S}$-recognizable for all words $v \in \N_{>0}^*$ and $P_v(Y)$ is nonempty for only finitely many $v$.

Our aim is to show that when $X$ is a subshift generated by an $\ans{S}$-automatic word, then $P_v(W(X))$ can be
described by a formula expressed in first order logic. When $\ans{S}$ is addable, the existence of such a formula
implies that $P_v(W(X))$ is $\ans{S}$-recognizable. By \autoref{lem:equivalence}, this means that $W(X)$ is weakly
$\ans{S}$-codable; see \autoref{thm:ws_weakly_codable}. Proving the existence of the logical formula for a general $v$
leads to a proof with complicated notation. We thus opt to prove the existence when $v = 111$ and $X$ is a binary
subshift. The proof has the flavor of the general proof, and it should convince the reader that the main ideas can be
carried out more generally.

Next we begin defining the formulas. We advise the reader to begin reading the proof of \autoref{prp:case_111} before
attempting to grasp the meaning of these formulas.

We now fix $v = 111$. Let $\infw{x}$ be an infinite binary word. Let $n_{000}$, $n_{001}$, $\ldots$, $n_{111}$ be
variables so that $n_{def}$ is one of these variables when $d, e, f \in \{0, 1\}$ are given. Let $a$, $b$, $c$ be free
variables. Define a formula $\varphi_{def}$ as follows:
\begin{equation*}
  \varphi_{def} = (\infw{x}[n_{def} + a] = d \land \infw{x}[n_{def} + b] = e \land \infw{x}[n_{def} + c] = f).
\end{equation*}
Moreover, we define formulas $\varphi_0(i)$, $\varphi_1(i)$, and $\varphi_2(i)$:
\begin{alignat*}{4}
  \varphi_0(i) &=\; &(&\infw{x}[n_{000} + i] = \infw{x}[n_{001} + i] = \infw{x}[n_{010} + i] = \infw{x}[n_{011} + i] = \\
               &    & &\infw{x}[n_{100} + i] = \infw{x}[n_{101} + i] = \infw{x}[n_{110} + i] = \infw{x}[n_{111} + i]), \\
  \varphi_1(i) &=\; &(&\infw{x}[n_{000} + i] = \infw{x}[n_{001} + i] = \infw{x}[n_{010} + i] = \infw{x}[n_{011} + i] \land \\
               &    & &\infw{x}[n_{100} + i] = \infw{x}[n_{101} + i] = \infw{x}[n_{110} + i] = \infw{x}[n_{111} + i]), \\
  \varphi_2(i) &=\; &(&\infw{x}[n_{000} + i] = \infw{x}[n_{001} + i] \land \infw{x}[n_{010} + i] = \infw{x}[n_{011} + i] \land \\
               &    & &\infw{x}[n_{100} + i] = \infw{x}[n_{101} + i] \land \infw{x}[n_{110} + i] = \infw{x}[n_{111} + i]).
\end{alignat*}
Let finally
\begin{align*}
  \psi(a, b, c) = ( (a < b < c) \land (\exists n_{000}, n_{001}, ..., n_{111}\colon & (\forall d, e, f \in \{0, 1\}\colon \varphi_{def}) \land \\
                                                                              &(\forall i \in [0, a)\colon \varphi_0(i)) \land \\
                                                                              &(\forall i \in (a, b)\colon \varphi_1(i)) \land \\
                                                                              &(\forall i \in (b, c)\colon \varphi_2(i)))).
\end{align*}

\begin{proposition}\label{prp:case_111}
  Let $X$ be the orbit closure of $\infw{x}$ in $\{0, 1\}^\N$. Then
  \begin{equation*}
    P_{111}(W(X)) = \{ (a, b, c) \in \N^3 : \psi(a, b, c)\}.
  \end{equation*}
\end{proposition}

\begin{proof}
  Suppose that there is a word $\infw{y}$ in the winning shift of $X$ such that it has letter $1$ at positions $a$,
  $b$, $c$ with $a < b < c$ and letter $0$ elsewhere so that $s(\infw{y}) = (a, b, c) \in P_{111}(W(X))$. By the
  definition of the winning shift, there exist words $w_{000}, w_{001}, \ldots, w_{111} \in \mathcal{L}(X)$, all of
  length $c+1$, such that the rooted partial binary tree formed by these words has outdegree $2$ at nodes on depths
  $a$, $b$, $c$, and outdegree $1$ elsewhere. This tree is the strategy tree for Alice's winning strategy. Since $X$ is
  over the alphabet $\{0,1\}$, at every branching we have one branch starting with $0$ and one starting with $1$. By
  ordering the words suitably, we may assume that
  \begin{equation*}
    w_{def} = u \cdot d \cdot u_d \cdot e \cdot u_{de} \cdot f
  \end{equation*}
  for all $d, e, f \in \{0, 1\}$, and for some words $u, u_d, u_{de}$, where $|u| = a$, $|u_d| = b-a-1$,
  $|u_{de}| = c-b-1$, and, as indicated by the subscripts, $u$ is the same for all $w_{def}$, $u_d$ only depends on $d$
  and $u_{de}$ only depends on $d$ and $e$.

  By the assumption that $X = \OC{\infw{x}}$, there exist $n_{def} \in \N$ such that
  $\infw{x}[n_{def}, n_{def+c}] = w_{def}$ for all $d, e, f \in \{0, 1\}$. These $n_{def}$ satisfy all of the
  $\forall$-statements in $\psi$. The first statement
  \begin{equation*}
    \forall d, e, f \in \{0, 1\}\colon (\infw{x}[n_{def} + a] = d \land \infw{x}[n_{def} + b] = e \land \infw{x}[n_{def} + c] = f)
  \end{equation*}
  holds because $w_{def}[a] = d \implies \infw{x}[n_{def} + a] = d$, so $\infw{x}[n_{def} + a] = d$ holds for all $d$,
  $e$, $f$. Similarly $\infw{x}[n_{def} + b] = e$ and $\infw{x}[n_{def} + c] = f$ hold by the definition of the words $w_{def}$.

  The second statement
  \begin{alignat*}{4}
    \forall i \in [0, a)&\colon &(&\infw{x}[n_{000} + i] = \infw{x}[n_{001} + i] = \infw{x}[n_{010} + i] = \infw{x}[n_{011} + i] = \\
                        &      & &\infw{x}[n_{100} + i] = \infw{x}[n_{101} + i] = \infw{x}[n_{110} + i] = \infw{x}[n_{111} + i])
  \end{alignat*}
  holds because $\infw{x}[n_{def} + i] = w_{def}[i] = u[i]$ for all $i \in [0, a)$ and $d,e,f \in \{0,1\}$, so all
  eight values $\infw{x}[n_{def} + i]$ are the same. The last two $\forall$-statements correspond to $u_d$ and $u_{de}$
  and are justified similarly.

  Conversely, if the $\forall$-statements hold for some choices of the $n_{def}$, then the words
  $\infw{x}[n_{def}, n_{def+c}]$ are of the desired form $w_{def} = u \cdot d \cdot u_d \cdot e \cdot u_{de} \cdot f$,
  and this proves that $(a, b, c) \in P_{111}(W(X))$.
\end{proof}

The following corollary can be deduced from the previous proposition using standard results concerning
$\ans{S}$-automatic sequences, but we provide the proof in our special case. See the discussion after
\autoref{thm:formula} for the full story.

\begin{corollary}\label{cor:words}
  Let $\ans{S}$ be an addable ANS, and let $\infw{x}$ in $\{0,1\}^\N$ be $\ans{S}$-automatic. Then
  $P_{111}(W(\OC{\infw{x}}))$ is $\ans{S}$-recognizable.
\end{corollary}
\begin{proof}
  Let $\ans{S} = (L, \prec)$. We explain the steps for constructing a finite automaton that accepts
  $\rep(P_{111}(W(\OC{\infw{x}})))$. Let $\mathcal{A}$ be a DFAO for $\infw{x}$. Let
  $(a, b, c) \in P_{111}(W(\OC{\infw{x}}))$. By the definition of an ANS, the condition $a < b < c$ is equivalent with
  \begin{equation*}
    \rep(a) \prec \rep(b) \prec \rep(c) \land \rep(a) \neq \rep(b) \land \rep(b) \neq \rep(c).
  \end{equation*}
  and this can be checked at the end. The check is possible because addability implies that $\ans{S}$ is comparable
  (i.e., the relation $\prec$ is rational) by \autoref{lem:addable_implications}.

  Consider a tuple
  \begin{equation*}
    (w_a, w_b, w_c, w_{000}, w_{001}, \ldots, w_{111})
  \end{equation*}
  and define $a = \val(w_a)$, $b = \val(w_b)$, $c = \val(w_c)$, and $n_{def} = \val(w_{def})$ for
  $d, e, f \in \{0, 1\}$. Since regular languages over a Cartesian product alphabet are closed under projection, we can
  eliminate existentially quantified words, and thus it is enough to show that the language of tuples such that $a$,
  $b$, $c$, $n_{def}$ satisfy $\psi$ is regular. Notice that $L$ is regular by addability, so we can check that all of
  the words are in $L$ by intersecting with the Cartesian product language $L^{11}$ (with padding).

  Using addability, we can construct an automaton $A_{d,e,f,0}$ such that $A_{d,e,f,0}$ accepts the word
  $(w_a,w_b,w_c,w_{000},...,w_{111})$ over $\{0,1,\#\}^{11}$ if and only if
  \begin{equation*}
    \infw{x}[n_{def} + a] = d.
  \end{equation*}
  Namely, $A_{d,e,f,0}$ is accepted by $\mathcal{A} \circ (+) \circ \pi_{0, 3+4d+2e+f}$ by taking the final states to
  be those where $\mathcal{A}$ outputs the symbol $d$, where $(+)\colon L^2 \to L$ is the sum transduction
  corresponding to $\ans{S}$ and $\pi_{i_1, \ldots,i_k}$ is the projection to components $i_1$, $\ldots$, $i_k$ of the
  alphabet. Similarly, we can construct automata $A_{d,e,f,1}$ and $A_{d,e,f,2}$ that accept when
  $\infw{x}[n_{def} + b] = e$ and when $\infw{x}[n_{def} + a] = d$, respectively.

  Now, let $L_0$ be the (finite) intersection of the languages of $A_{d,e,f,i}$ for $d,e,f \in \{0,1\}$,
  $i \in \{0,1,2\}$. The language $L_0$ is regular by closure properties of regular languages. Clearly $L_0$ contains
  precisely those tuples $(w_a,w_b,w_c,w_{000},\ldots,w_{111})$ such that the corresponding numbers $a$, $b$, $c$,
  $n_{def}$, $d,e,f \in \{0,1\}$, satisfy the first $\forall$-statement of $\psi$.

  The other three $\forall$-statements are actual $\forall$-quantifications (rather than shorthands for a finite
  conjunction). We explain how to handle the middle one
  \begin{alignat*}{4}
    \forall i \in (a, b)&\colon &(&\infw{x}[n_{000} + i] = \infw{x}[n_{001} + i] = \infw{x}[n_{010} + i] = \infw{x}[n_{011} + i] \land \\
                        &       & &\infw{x}[n_{100} + i] = \infw{x}[n_{101} + i] = \infw{x}[n_{110} + i] = \infw{x}[n_{111} + i]),
  \end{alignat*}
  the others being similar. We first define an auxiliary automaton that accepts those $12$-tuples
  $(w_a,w_b,w_c,w_{000},...,w_{111},w_i)$ such that the formula holds with $i = \val(w_i)$ and then $\forall$-eliminate
  $w_i$. To accept such $12$-tuples, observe again that we can compute the value $\infw{x}[n_{000} + i]$ by an
  automaton (just like in the definition of $A_{d,e,f,j}$). Thus, we can also perform a comparison of these finitely
  many values, and compute the conjunction $\land$ of these values at the end of the computation. Then
  $(w_a,w_b,w_c,w_{000},...,w_{111},w_i)$ is accepted in the correct situations whenever $\rep(w_i) \in (a, b)$.

  Since we want to eliminate a $\forall$-quantifier, we want the automaton to accept all other values of $i$. For this,
  we use the rationality of $\prec$. The automaton can simply check whether $w_i \in L$ and
  $w_a \prec w_i \prec w_b, w_i \neq w_a, w_i \neq w_b$ holds, and accept if this is not the case. Elimination of
  $\forall$-quantifiers is dual to elimination of $\exists$-quantifiers. Simply complement, project to all but the
  $12$th coordinate, and complement again.

  Now, let $L_0$, $L_1$, $L_2$, and $L_3$ be the four regular sublanguages of $L^{11}$ corresponding to the
  $\forall$-quantifiers. Then $\pi_{0,1,2}(L_0 \cap L_1 \cap L_2 \cap L_3)$ is the desired regular language.
\end{proof}

As indicated above, \autoref{prp:case_111} and the previous corollary are straightforward to generalize. The assumption
of a binary alphabet has no effect on the formulas in \autoref{prp:case_111}, but there are some minor changes in the
proof ($d$, $e$, $f$ should range over some cardinality two subsets of the alphabet). To cover a general word $v$, one
should draw a directed tree where all nodes on depth $i$ have out-degree $v_i+1$. These nodes can be indexed naturally
by words $u$ such that $u \leq v$ (with elementwise comparison), and using a variable $n_u$ for each such $u$ one can
then easily program the analogue $\varphi_v$ of the formula $\varphi_{def}$ by writing out in first-order logic that
the words starting at the positions $n_u$ form a strategy tree.

We obtain the following result.

\begin{theorem}\label{thm:formula}
  If $X$ is a subshift generated by $\infw{x}$ and $v$ is a word of length $k$ in $\N_{>0}^*$, then there exists a
  first-order formula $\psi(n_1, \ldots, n_k)$ on $\ans{S}$-recognizable predicates such that
  \begin{equation*}
    P_v(W(X)) = \{ (n_1, \ldots, n_k) \in \N^k : \psi(n_1, \ldots, n_k)\}.
  \end{equation*}
\end{theorem}

The conclusion of $\ans{S}$-recognizability of \autoref{cor:words} is obtained whenever a predicate is formed from
$\ans{S}$-recognizable predicates using the logical connectives $\land$, $\lor$, $\lnot$, $\implies$, $\iff$ and the
quantifiers $\forall$, $\exists$ on variables describing elements of $\N$. In the case of the usual base-$b$ numeration
system, this is the famous B\"{u}chi-Bruy\`{e}re Theorem (see \cite{1994:logic_and_p-recognizable_sets_of_integers} for
a proof). The extension to general ANS is straightforward and is sketched in
\cite[Sect.~7.3]{2022:regular_sequences_and_synchronized_sequences_in_abstract}. The important point here is that
addability and existence of an automaton for $\infw{x}$ ensures that the predicates $\varphi_{def}$, $\varphi_0$,
$\varphi_1$, and $\varphi_2$ are $\ans{S}$-recognizable.

\begin{theorem}\label{thm:ws_weakly_codable}
  Let $\ans{S}$ be an addable ANS. Let $\infw{x}$ be an $\ans{S}$-automatic word and $X$ its orbit closure. Then the
  winning shift $W(X)$ of $X$ is weakly $\ans{S}$-codable.
\end{theorem}
\begin{proof}
  By \autoref{lem:equivalence}, it suffices to prove that $P_v(W(X))$ is $\ans{S}$-recognizable for all
  $v \in \N_{>0}^*$. By \autoref{thm:formula}, the set $P_v(W(X))$ is defined by a first-order formula $\psi$ on
  $\ans{S}$-recognizable predicates. Similar to \autoref{cor:words}, it follows from the theory of $\ans{S}$-automatic
  sequences with addable $\ans{S}$ that $\rep(P_v(W(X)))$ is regular.
\end{proof}

\section{\texorpdfstring{$\ans{S}$}{S}-codable Winning Shifts}\label{sec:Concrete}
In this section, we consider conditions that ensure that a winning shift of a subshift generated by an
$\ans{S}$-automatic word is $\ans{S}$-codable. A factor $w$ of a word $\infw{x}$ is \emph{right special} if
$wa, wb \in \Lang[\infw{x}]{}$ for distinct letters $a$ and $b$.

%\begin{proposition}\label{prp:hereditary_fcd}
%  Let $X$ be a hereditary subshift contained in $\N^\N$. If $X$ has sublinear factor complexity and the shift map
%  $T\colon X \to X$ is surjective, then $X$ is countable and has finite coding dimension.
%\end{proposition}
%\begin{proof}
%  Say a word $\infw{x}$ in $X$ contains $M$ occurrences of the letter $1$. By downgrading, we may assume that all other
%  letters of $\infw{x}$ equal $0$. By downgrading the letters $1$ in $\infw{x}$ one by one, we find factors $10^{k_i}1$
%  for at least $M$ distinct integers $k_i$. By surjectivity and downgrading, we find in $X$ words
%  $0^\ell 10^{k_i}1 0^\omega$ for arbitrarily large $\ell$. For large enough $n$, each such word (for a suitable
%  $\ell$) contains $n - k_i + 1$ factors containing $10^{k_i}1$. Since the integers $k_i$ are distinct, we find that
%  the factor complexity of $X$ is at least $Mn + M - \sum_i k_i$ for $n$ large enough. We deduce that if the factor
%  complexity of $X$ is at most $Kn$ for $n$ large enough, then a word in $X$ contains at most $K$ letters $1$.
%  Sublinear factor complexity implies that words of $X$ contain finitely many distinct letters, and the conclusion
%  follows. Notice that a subshift with finite coding dimension is necessarily countable.
%\end{proof}

\begin{proposition}\label{prp:minimal_linear_fcd}
  If $X$ is a transitive subshift with sublinear factor complexity, then $W(X)$ is countable and has finite coding
  dimension.
\end{proposition}
\begin{proof}
  Let $X$ be a transitive subshift with sublinear factor complexity. Let $s_{\infw{x}}(n)$ be the number of right
  special factors of length $n$ in $\Lang[\infw{x}]{}$ for $\infw{x} \in X$. Assume that $W(X)$ has infinite coding
  dimension. Since $W(X)$ is hereditary, there exists a strictly increasing sequence $(B_i)$ of positive integers such
  that there exists $\infw{x}_i \in W(X)$ with $\sum \infw{x}_i = B_i$ for all $i$. Moreover, we may assume that each
  $\infw{x}_i$ contains only the letters $0$ and $1$. Each strategy tree corresponding to $\infw{x}_i$ branches exactly
  $B_i$ times. Clearly the final branching of a strategy tree corresponds to $2^{B_i - 1}$ right special words of a
  common length in the language of $X$. Since $(B_i)$ is unbounded, it follows by transitivity that
  $s_{\infw{y}}(n)$ is unbounded for some $\infw{y} \in X$. A famous result of Cassaigne
  \cite[Thm.~1]{1996:special_factors_of_sequences_with_linear_subword}, \cite[Thm.~4.9.3]{2010:combinatorics_automata_and_number_theory}
  states that the infinite word $\infw{y}$ has sublinear factor complexity if and only $s_{\infw{y}}(n)$ is bounded.
  Therefore $\infw{y}$ has superlinear factor complexity, and we conclude that $X$ has superlinear factor complexity as
  well. If $W(X)$ has finite coding dimension, then $W(X)$ is clearly countable. The claim follows.
\end{proof}

Putting \autoref{prp:minimal_linear_fcd} and \autoref{thm:ws_weakly_codable} together gives the following result which
is the main result of the paper.

\begin{theorem}\label{thm:main}
  Let $\ans{S}$ be an addable ANS. Suppose that $X$ is a subshift generated by an $\ans{S}$-automatic word
  having sublinear complexity. Then $W(X)$ is $\ans{S}$-codable.
\end{theorem}

We have the following immediate corollary for Pisot numeration systems. Notice that the usual base-$b$ numeration
system is a Pisot numeration system.

\begin{corollary}\label{cor:pisot}
  Let $\ans{S}$ be a Pisot numeration system. If $X$ is a subshift generated by an $\ans{S}$-automatic word,
  then $W(X)$ is $\ans{S}$-codable.
\end{corollary}
\begin{proof}
  Let $\infw{x}$ be an $\ans{S}$-automatic word. The subshift generated by $\infw{x}$ is transitive by definition. By
  \cite[Thm.~3.4]{2019:automatic_sequences_based_on_parry_or_bertrand} all Parry-automatic words have sublinear factor
  complexity. Since all Pisot numeration systems are Parry numeration systems
  \cite[Remark~3]{2019:automatic_sequences_based_on_parry_or_bertrand}, it follows that $\infw{x}$ has sublinear factor
  complexity. It is a well-know fact that Pisot numeration systems are addable
  \cite{1992:representations_of_numbers_and_finite_automata}, so the claim follows from \autoref{thm:main}.
\end{proof}

We note that we cannot prove a result analogous to \autoref{cor:pisot} for Parry-automatic
words since addition is not always rational in a Parry numeration system. See
\cite[Ex.~3]{1997:on_the_sequentiality_of_the_successor_function}.

\begin{lemma}\label{lem:hereditary_low_complexity}
  There exists a hereditary subshift $X$ with sublinear factor complexity such that $X$ has infinite coding dimension.
\end{lemma}
\begin{proof}
  Let $(n_i)$ be a sequence of positive integers, and set $m_j = \sum_{i=1}^j (n_i + 1)$ for $j \geq 1$. In order to
  make the arguments below work out, we choose the sequence $(n_i)$ to satisfy $m_{j-1} = \mathcal{O}(\log n_j)$. Let
  $\infw{x} = \prod_{i=1}^\infty 0^{n_i}1$, and let $X$ be the hereditary closure of the subshift generated by
  $\infw{x}$. It is clear that $X$ has infinite coding dimension, so it suffices to show that $X$ has sublinear factor
  complexity.

  Suppose that $n$ is such that $n_j + 2 \leq n \leq n_{j+1} + 1$ for some $j \geq 2$. When $p \geq m_{j-1}$, the
  factor $\infw{x}[p, p+n-1]$ contains at most one letter $1$, so these positions contribute a total of $n+1$ to the
  factor complexity of $\infw{x}$. Moreover, downgrading a letter (using hereditarity) does not increase the number of
  factors. Therefore we need to show that the contribution from the earlier positions, including downgrading, is
  $\mathcal{O}(n)$. Before downgrading, we have at most $m_{j-1}$ distinct factors of length $n$. Since $n \leq n_{j+1}
  + 1$, each early factor contains at most $j - 1$ letters $1$. Downgrading each letter independently thus produces at
  most $\smash[t]{\sum_{i = 1}^{j-1} 2^i = 2^j - 2}$ new factors. Now $\log \log n_j = \Omega(\log m_{j-1})$. Since
  $m_j$ clearly grows at least exponentially, we find that $\log m_{j-1} = \Omega(j)$, so
  $j = \mathcal{O}(\log \log n_j)$. As $n_j < n$, we thus have $j = \mathcal{O}(\log \log n)$. Since
  $\smash[t]{n + 1 + (2^j - 2)m_{j-1} = n + 1 + \mathcal{O}(\log^2 n)}$, we thus have $\mathcal{O}(n)$ factors of
  length $n$.
\end{proof}

A simple argument using the pumping lemma shows that the word in the proof of \autoref{lem:hereditary_low_complexity}
is not Pisot-automatic.

Next we provide an example of a substitutive subshift associated with a regular ANS which does not have finite coding
dimension. This example was suggested to us by J.\ Cassaigne. We show below that the associated ANS is not addable, so
we cannot deduce that the conclusion of \autoref{thm:main} fails without the assumption of sublinear factor complexity.

Let $\sigma$ be the substitution defined by $a \mapsto abab$, $b \mapsto b$. Let
\begin{equation*}
  \infw{z} = ababbababbbababbababbbbababbababbbababbababbbbbababbababbbababbababbbb
  \dotsm
\end{equation*}
be its infinite fixed point and $Z = \OC{\infw{z}}$. It follows from
\cite[Thm.~4.7.66]{2010:combinatorics_automata_and_number_theory} that the factor complexity of $\infw{z}$ is in
$\Theta(n^2)$. (For a direct proof of the lower bound, consider words of the form $uab^nav$ where $b^n$ covers the
central position, and observe that one may choose the positions of the two $a$'s freely to obtain $m^2$ words of length
$2m$.) Our aim is to prove the following result.

\begin{proposition}\label{prp:z_icd}
  The winning shift $W(Z)$ of $Z$ has infinite coding dimension.
\end{proposition}

For this, we need information on right special factors in $\Lang[Z]{}$ and the following general lemma.

\begin{lemma}\label{lem:rs_conditions}
  Let $X$ be a subshift that satisfies the following properties:
  \begin{enumerate}[(i)]
    \item Whenever $w$ is a long enough right special factor in $\Lang[X]{}$, there exists a letter $a$ such that $wa$ is
          right special in $\Lang[X]{}$.
    \item Right special words are dense in $X$ meaning that each word in $X$ is a limit of a sequence of right special
          words of $\Lang[X]{}$.
  \end{enumerate}
  Then $W(X)$ has infinite coding dimension.
\end{lemma}
\begin{proof}
  Suppose that $\infw{x} = u0^\omega$ and $\infw{x} \in W(X)$. The condition (ii) implies that there are arbitrarily
  long right special words in $\Lang[X]{}$, so we may assume that $\sum u > 0$. Consider the strategy tree
  corresponding to $u$. Let $w_0$, $\ldots$, $w_{k-1}$ be the words in $\Lang[X]{}$ that correspond to paths from the
  root to the leaves. The condition (ii) guarantees that each $w_i$ can be extended to a right special factor $u_i$.
  Moreover, the condition (i) implies that each $u_i$ can be extended to a right special factor of length $\ell$ where
  $\ell$ does not depend on $i$. Consider the words of length $|v_i| + 1$ obtained from the words $v_i$ by extending
  each $v_i$ in at least two ways. By considering the tree corresponding to these words, we find that it is a strategy
  tree for $u 0^{\ell - |u|} 1$. In other words, the word $u 0^{\ell - |u|} 10^\omega$ belongs to $W(X)$. By repeating
  the argument on this word, we find words in $W(X)$ with arbitrarily large sums. The claim follows.
\end{proof}

It is easy to show that $\infw{z} = \lim_k p_k$ where $p_1 = a$ and $p_{k+1} = p_k b^k p_k$. From this, it is
straightforward to deduce the following.

\begin{lemma}\label{lem:rs_characterization}
  The right special factors in $\Lang[Z]{}$ ending with $ab^k$ are exactly the suffixes of $b^k p_k b^k$ for
  $k \geq 1$.
\end{lemma}

\begin{proof}[Proof of \autoref{prp:z_icd}]
  It suffices to verify the conditions of \autoref{lem:rs_conditions}. Let $w$ be a right special factor in
  $\Lang[Z]{}$. For the condition (i), it suffices to prove that $wb$ is also right special in $\Lang[Z]{}$. If
  $w = b^i$, then it is clear that $w$ and $wb$ are right special. Otherwise there exists $k \geq 1$ such that the word
  $w$ is a suffix of $b^k p_k b^k$ by \autoref{lem:rs_characterization}. The word $b^{k+1} p_{k+1} b^{k+1}$ is right
  special according to \autoref{lem:rs_characterization}, and it has suffix $b^k p_k b^{k+1}$ meaning that $wb$ is
  right special.

  The condition (ii) is immediate as $Z = \OC{\infw{z}}$ and $\infw{z}$ has arbitrarily long right special prefixes by
  \autoref{lem:rs_characterization}.
\end{proof}

The sums grow slowly, the shortest factor in $W(Z)$ with sum at least $4$ is $1010^410^{197}1$.

Let us then describe an ANS $\ans{Z}$ associated with $\sigma$ and $\infw{z}$. We set $\rep[S](0) = \varepsilon$. Let
$n$ be a positive integer and $\ell$ be the smallest integer such that
\begin{equation*}
  |\sigma^\ell(a)| \leq n < |\sigma^{\ell+1}(a)| = |\sigma^\ell(aba)| + 1.
\end{equation*}
Define
\begin{align*}
  m &= \begin{cases}
         |\sigma^\ell(a)|, &\text{if $n = |\sigma^\ell(a)|$}, \\
         |\sigma^\ell(ab)|    , &\text{if $|\sigma^\ell(ab)| \leq n < |\sigma^\ell(aba)|$}, \\
         |\sigma^\ell(aba)|        , &\text{if $n = |\sigma^\ell(aba)|$}
       \end{cases}
  \quad \text{and} \\
  \delta &= \begin{cases}
              1, &\text{if $n = |\sigma^\ell(a)|$}, \\
              2    , &\text{if $|\sigma^\ell(ab)| \leq n < |\sigma^\ell(aba)|$}, \\
              3        , &\text{if $n = |\sigma^\ell(aba)|$}.
           \end{cases}
\end{align*}
We define recursively $\rep[S](n) = \delta \alpha$ where $\alpha$ is the word $\rep[S](n - m)$ with enough initial
zeros to make $\rep[S](n)$ a word of length $\ell + 1$. What happens here is that we find which of the words
$\sigma^\ell(a)$, $\sigma^\ell(ab)$, $\sigma^\ell(aba)$ is the longest prefix of the prefix $p$ of $\infw{z}$ of length
$n$, write a symbol in $\{1,2,3\}$ according to the word, and repeat the procedure for the remaining suffix of $p$. The
remaining suffix is a prefix of $\infw{z}$ because the square of the word $\sigma^\ell(a)$ is always a prefix of
$\infw{z}$. The representations of the integers $1$ to $22$ are given in \autoref{tbl:rep}.

\begin{table}
\centering  
\begin{tabular}{| c | c | c || c | c | c |}
  \hline
  $n$  & $\rep(n)$ & \text{prefix}   & $n$ & $\rep(n)$ & \text{prefix} \\ \hline
  $1$  & $1$       & $a$             & $12$ & $201$    & $\sigma^2(ab)a$ \\
  $2$  & $2$       & $ab$            & $13$ & $202$    & $\sigma^2(ab)ab$ \\
  $3$  & $3$       & $aba$           & $14$ & $203$    & $\sigma^2(ab)aba$ \\
  $4$  & $10$      & $\sigma(a)$     & $15$ & $210$    & $\sigma^2(ab)\sigma(a)$ \\
  $5$  & $20$      & $\sigma(ab)$    & $16$ & $220$    & $\sigma^2(ab)\sigma(ab)$ \\
  $6$  & $21$      & $\sigma(ab)a$   & $17$ & $221$    & $\sigma^2(ab)\sigma(ab)a$ \\
  $7$  & $22$      & $\sigma(ab)ab$  & $18$ & $222$    & $\sigma^2(ab)\sigma(ab)ab$ \\
  $8$  & $23$      & $\sigma(ab)aba$ & $19$ & $223$    & $\sigma^2(ab)\sigma(ab)aba$ \\
  $9$  & $30$      & $\sigma(aba)$   & $20$ & $230$    & $\sigma^2(ab)\sigma(aba)$ \\
  $10$ & $100$     & $\sigma^2(a)$   & $21$ & $300$    & $\sigma^2(aba)$ \\
  $11$ & $200$     & $\sigma^2(ab)$  & $22$ & $1000$   & $\sigma^3(a)$ \\
  \hline
\end{tabular}
\caption{Representations of the numbers $1$, $\ldots$, $22$ in the ANS associated with the substitution $\sigma$.}\label{tbl:rep}
\end{table}

The numeration system $\ans{Z}$ is the Dumont-Thomas numeration system associated with the substitution $\sigma$. See
\cite[Sect.~2.1]{2014:formal_languages_automata_and_numeration_systems_2} for more details about these numeration
systems. Notice that the language $\rep[S](\N)$ is regular as it has the regular expression
\begin{equation*}
  2(0 + 2)^* + 2(0 + 2)^*(1 + 3) 0^* + 3 0^* + 1 0^*. 
\end{equation*}
Moreover, the order relation $\prec$ is given by the radix order based on the total order $0 \prec 1 \prec 2 \prec 3$
on the alphabet.

The fixed point $\infw{z}$ is a $\ans{Z}$-automatic word. If we use the convention that the first letter of $\infw{z}$
equals the output of an automaton fed with $\rep[S](1)$, this is easy to see by consulting \autoref{tbl:rep}. Indeed,
the words $\sigma(a)$ and $\sigma(b)$ both end with $b$, so $a$ needs to be outputted exactly when the representation
ends with $1$ or $3$. The word $\infw{z}$ is still $\ans{Z}$-automatic if we index from $0$. While a DFAO could be
written for $\infw{z}$ in this case, it is easiest to make the conclusion using well-known closure properties of
$\ans{Z}$-automatic words; see \cite[Prop.~14]{2002:more_on_generalized_automatic_sequences}. The point is that the
successor function maps regular sets to regular sets; see the proof of
\cite[Prop.~16]{2001:numeration_systems_on_a_regular_language}.

Combining the above with \autoref{prp:z_icd}, we obtain the following result.

\begin{proposition}
  There exists a comparable ANS $\ans{S}$ and an $\ans{S}$-automatic word $\infw{x}$ such that $W(\OC{\infw{x}})$ has
  infinite coding dimension.
\end{proposition}

Let us next prove that the ANS $\ans{Z}$ is not addable.

\begin{proposition}
  The ANS $\ans{Z}$ is not addable.
\end{proposition}
\begin{proof}
  To simplify notation, we write $u + v$ in place of $\rep(\val(u) + \val(v))$ when $u, v \in \rep(\N)$. Notice that if
  $n > 0$, then
  \begin{equation*}
    2^n30^{m-n} 0^{\ell} + 1 = 2^{n-1}30^{m-n+1} 0^\ell
  \end{equation*}
  since there is no word in $\rep(\N)$ strictly between the two in lexicographic order. Now consider $w$ in $\rep(\N)$
  such that $|w| \leq \ell$ and $\val(w) > n+1$. Then by induction we have
  \begin{align*}
    2^n30^{m-n} 0^{\ell} + w &= 30^m 0^{\ell} + \rep(\val(w) - n) \\
    &= 100^m 0^{\ell} + \rep(\val(w) - n - 1) \\
    &= 200^m 0^{\ell} + \rep(\val(w) - n - 2).
  \end{align*}
  By the special form of the language, the left quotient $200^m \setminus \rep(\N)$ is equal to $\rep(\N)$, from which
  we obtain that if $u = \rep(\val(w) - n - 2)$ then
  \begin{equation*}
    200^m 0^{\ell} + u = 200^m 0^{\ell - |u|} u.
  \end{equation*}

  Assume for a contradiction that $\ans{Z}$ is addable, that is, assume that the addition relation $R$, which is a
  subset of $\rep(\N)^3$, is regular and accepted by an automaton $\mathcal{A}$ with $Q$ states. Pick $m > Q$,
  $\val(w) > m + 1$, and $\ell \geq |w|$. Then by the above argument, $R$ contains the triple
  \begin{equation*}
    \begin{bmatrix}
      2^n 3 0^{m-n} 0^\ell \\
      w \\
      200^m 0^{\ell - |\rep(\val(w) - n - 2)|} \rep(\val(w) - n - 2)
    \end{bmatrix}^{\#}
  \end{equation*}
  for each $n$ such that $0 \leq n \leq m$.

  By the pigeonhole principle, there exist $n_1$ and $n_2$ such that $0 \leq n_1 < n_2 \leq m$ and in some accepting runs for
  \begin{equation*}
    t_1 = \begin{bmatrix}
            2^{n_1} 3 0^{m-n_1} 0^{\ell} \\
            w \\
            200^m 0^{\ell - |\rep(\val(w) - n_1 - 3)|} \rep(\val(w) - n_1 - 2)) 
          \end{bmatrix}^{\#}
  \end{equation*}
  and
  \begin{equation*}
    t_2 = \begin{bmatrix}
            2^{n_2} 3 0^{m-n_2} 0^{\ell} \\
            w \\
            200^m 0^{\ell - |\rep(\val(w) - n_2 - 3)|} \rep(\val(w) - n_2 - 2)
          \end{bmatrix}^{\#}
  \end{equation*}
  the automaton $\mathcal{A}$ is in the same state $q$ after reading the prefix of length $m+2$. This means that also
  \begin{equation*}
    t = \begin{bmatrix}
          2^{n_1} 3 0^{m-n_1} 0^{\ell} \\
          w \\
          200^m 0^{\ell - |\rep(\val(w) - n_2 - 3)|} \rep(\val(w) - n_2 - 2)
        \end{bmatrix}^{\#}
  \end{equation*}
  is in $R$ since we can use the accepting run for $t_1$ for $m+2$ steps and then use the accepting run for $t_2$ for
  the remaining steps. However, $t \notin R$, since
  \begin{align*}
    2^{n_1} 3 0^{m-n_1} 0^{\ell} + w &= 200^m 0^{\ell - |\rep(\val(w) - n_1 - 3)|} \rep(\val(w) - n_1 - 3)) \\
                                     &\neq 200^m 0^{\ell - |\rep(\val(w) - n_2 - 3)|} \rep(\val(w) - n_2 - 3)).
  \end{align*}
  This is a contradiction.
\end{proof}

\section{Automata for the Winning Shifts of Certain Automatic Words}\label{sec:automata}
\autoref{thm:main} is constructive: it is in principle possible to find an automaton accepting a coding of $W(X)$.
There are software packages like Walnut \cite{2016:automatic_theorem_proving_in_walnut} that can do this automatically.
In our attempt to directly input the formula of \autoref{thm:formula}, Walnut quickly ran out of memory even in the
case of one of the simplest automatic words, the Thue-Morse word. Using a logically equivalent form of the formulas,
the computation becomes manageable.

Let $\infw{x}$ be a binary $\ans{S}$-automatic word for an addable ANS $\ans{S}$ and $X$ the subshift generated by
$\infw{x}$. Let $\operatorname{factorEq}(i,n,m)$ be a predicate that is true if and only if
$\infw{x}[n,n+i-1] = \infw{x}[m,m+i-1]$. Let $\operatorname{isRS}(i,n)$ be a predicate that is true if and only if the
factor $\infw{x}[n,n+i-1]$ is right special. Consider a predicate $\operatorname{extRS2}(i,j,n)$ defined by the formula
\begin{alignat*}{4}
  \operatorname{extRS2}(i,j,n) & & = i < j \land \exists m_1,m_2\colon &(\operatorname{isRS}(j,m_1) \land \operatorname{isRS}(j,m_2) \land \\
                               & & & \operatorname{factorEq}(i,m_1,m_2) \land \operatorname{factorEq}(i,n,m_1) \land \\
                               & & & \infw{x}[m_1+i] \neq \infw{x}[m_2+i]).
\end{alignat*}
The predicate is true if and only if at position $n$ of $\infw{x}$ there is a right special factor $w$ of length $i$
such that both $w0$ and $w1$ can be extended to a right special factor of length $j$. In other words, the predicate
$\exists n \operatorname{extRS2}(i,j,n)$ is true if and only if there is a strategy tree with branchings at positions
$i$ and $j$. The latter is equivalent with $0^i 1 0^{j-i-1} 1 0^\omega \in W(X)$.

The formula
\begin{alignat*}{3}
  \operatorname{extRS3}(i,j,k,n_1,n_2) & & & = i < j \land \operatorname{extRS2}(j,k,n_1) \land \operatorname{extRS2}(j,k,n_2) \land \\
                                       & & & \hspace{1.2em} \operatorname{factorEq}(i,n_1,n_2) \land \infw{x}[n_1+i] \neq \infw{x}[n_2+i]
\end{alignat*}
encodes strategy trees with three branchings. Here $n_1$ and $n_2$ are positional variables that indicate starting
positions of right special factors in $\infw{x}$. There are two of them because for three branchings there are two
right special factors of a common length whose extensions can all be extended to right special factors of a common
length. The following formula encodes strategy trees with four branchings:
\begin{alignat*}{3}
  \operatorname{extRS4}(i,j,k,\ell,n_1,n_2,n_3,n_4) & & & = i < j \land \operatorname{extRS3}(j,k,\ell,n_1,n_2) \land \operatorname{extRS3}(j,k,\ell,n_3,n_4) \land \\
                                                    & & & \hspace{1.2em} \operatorname{factorEq}(i,n_1,n_2) \land \operatorname{factorEq}(i,n_2,n_3) \land \operatorname{factorEq}(i,n_3,n_4) \land \\
                                                    & & & \hspace{1.2em} \infw{x}[n_1+i] = \infw{x}[n_2+i] \land \infw{x}[n_2+i] \neq \infw{x}[n_3+i] \land \\
                                                    & & & \hspace{1.2em} \infw{x}[n_3+i] = \infw{x}[n_4+i].
\end{alignat*}
The advantage of this formulation is that the number of positional free variables is cut to one quarter. It also helps
that we can directly use previously generated automata for $\operatorname{extRS1}$, $\operatorname{extRS2}$, $\ldots$
instead of encoding all previous steps into a single formula.

\begin{figure}
  \centering
  \includegraphics[width=0.7\textwidth]{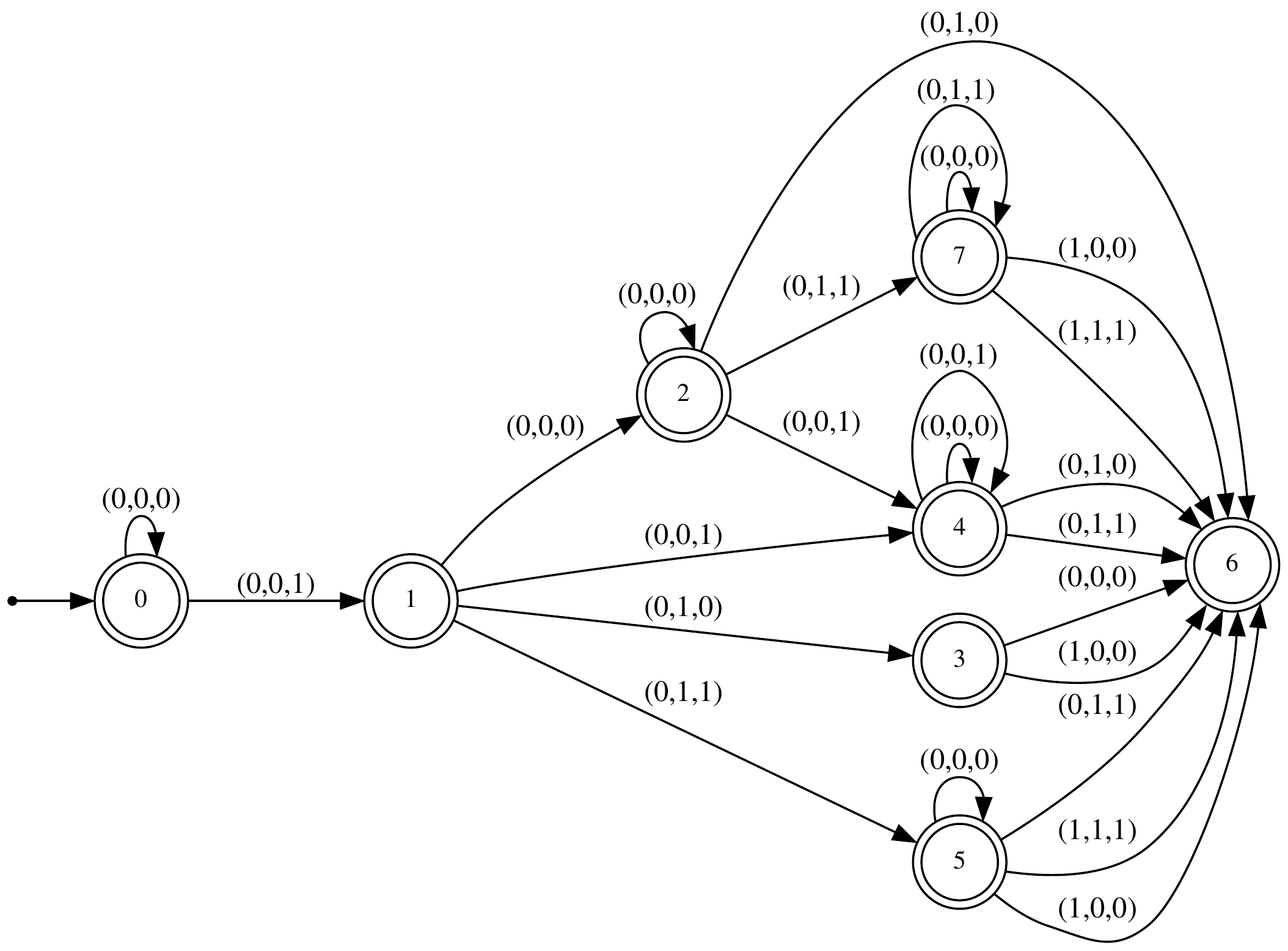}
  \caption{An automaton accepting the winning shift of the Thue-Morse word.}\label{fig:tm}
\end{figure}

Let now $\infw{x}$ be the Thue-Morse word, the fixed point $\mu^\omega(0)$ of the substitution
$\mu\colon 0 \mapsto 01, 1 \mapsto 10$. It is straightforward to prove that the $n$th letter of $\infw{t}$ equals the
number of $1$'s in the binary representation of $n$ modulo $2$, so $\infw{x}$ is a $2$-automatic word. The following
expresses the above predicates in Walnut's syntax for $\infw{x}$.
\begin{verbatim}
def factorEq "Ai (0 <= i & i < k) => T[n+i] = T[m+i]":
def isRS "Em1,m2 $factorEq(k,n,m1) & $factorEq(k,n,m2) & T[m1+k] != T[m2+k]":

def extRS1 "En $isRS(i,n)":
def extRS2 "i < j & Em1,m2 $isRS(j,m1) & $isRS(j,m2) &
                           $factorEq(i,m1,m2) & $factorEq(i,n,m1) &
                           T[m1+i] != T[m2+i]":
def extRS3 "i < j & $extRS2(j,k,n1) & $extRS2(j,k,n2) & $factorEq(i,n1,n2) &
                    T[n1+i] != T[n2+i]":
def extRS4 "i < j & $extRS3(j,k,l,n1,n2) & $extRS3(j,k,l,n3,n4) &
                    $factorEq(i,n1,n2) & $factorEq(i,n2,n3) &
                    $factorEq(i,n3,n4) &
                    T[n1+i] = T[n2+i] & T[n2+i] != T[n3+i] & T[n3+i] = T[n4+i]":
\end{verbatim}
When we input
\begin{verbatim}
def L4 "En1,n2,n3,n4 $extRS4(i,j,k,l,n1,n2,n3,n4)":
\end{verbatim}
to Walnut, we obtain an automaton that accepts $\rep(i,j,k,\ell)$ for those $4$-tuples $(i,j,k,\ell)$ for which the
predicate $\exists n_1, n_2, n_3, n_4 \operatorname{extRS4}(i,j,k,\ell,n_1,n_2,n_3,n_4)$ is true. This automaton in
fact rejects its all inputs, so there is no strategy tree with four branchings in $X$. Analogous computation for
strategy trees with three branchings produces an automaton that accepts some $3$-tuples $(i,j,k)$. This means that the
coding dimension of $X$ is $3$. Both computations finished in a few seconds.

\begin{figure}
  \centering
  \includegraphics[width=0.7\textwidth]{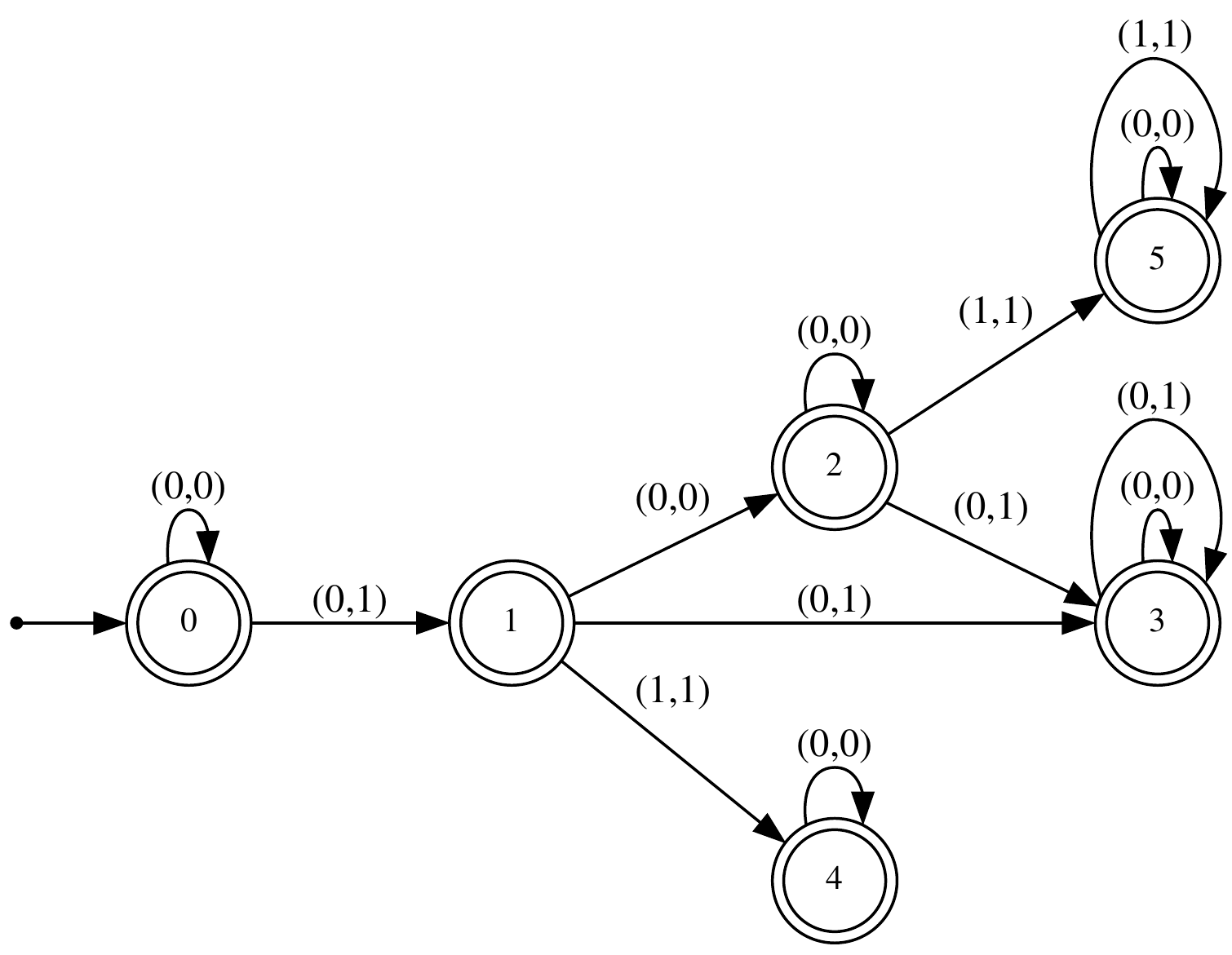}
  \caption{An automaton accepting the winning shift of the period-doubling word.}\label{fig:pd}
\end{figure}

\begin{figure}
  \centering
  \includegraphics[width=\textwidth]{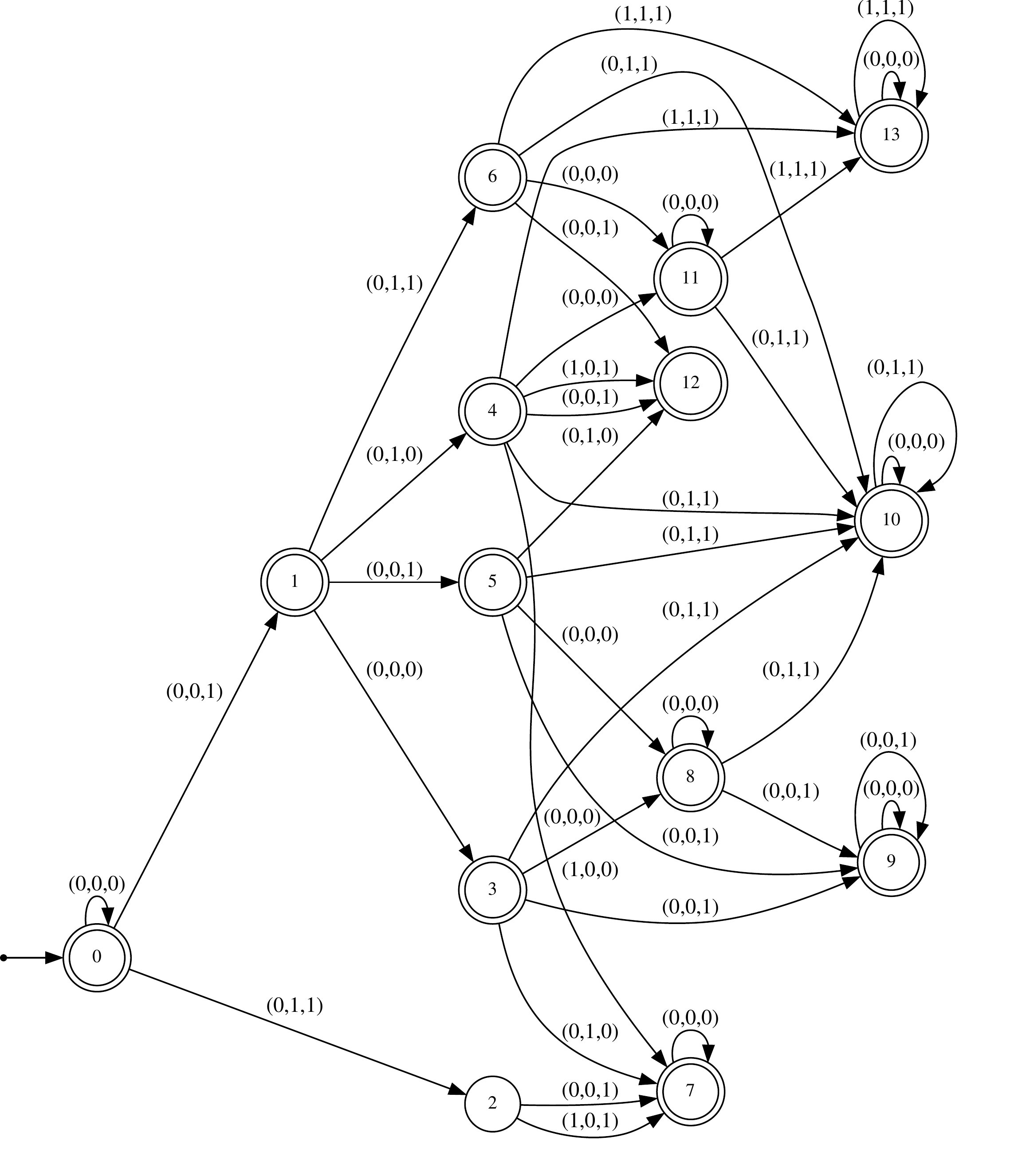}
  \caption{An automaton accepting the winning shift of the regular paperfolding word.}\label{fig:paper}
\end{figure}

\begin{figure}
  \centering
  \includegraphics[width=0.8\textwidth]{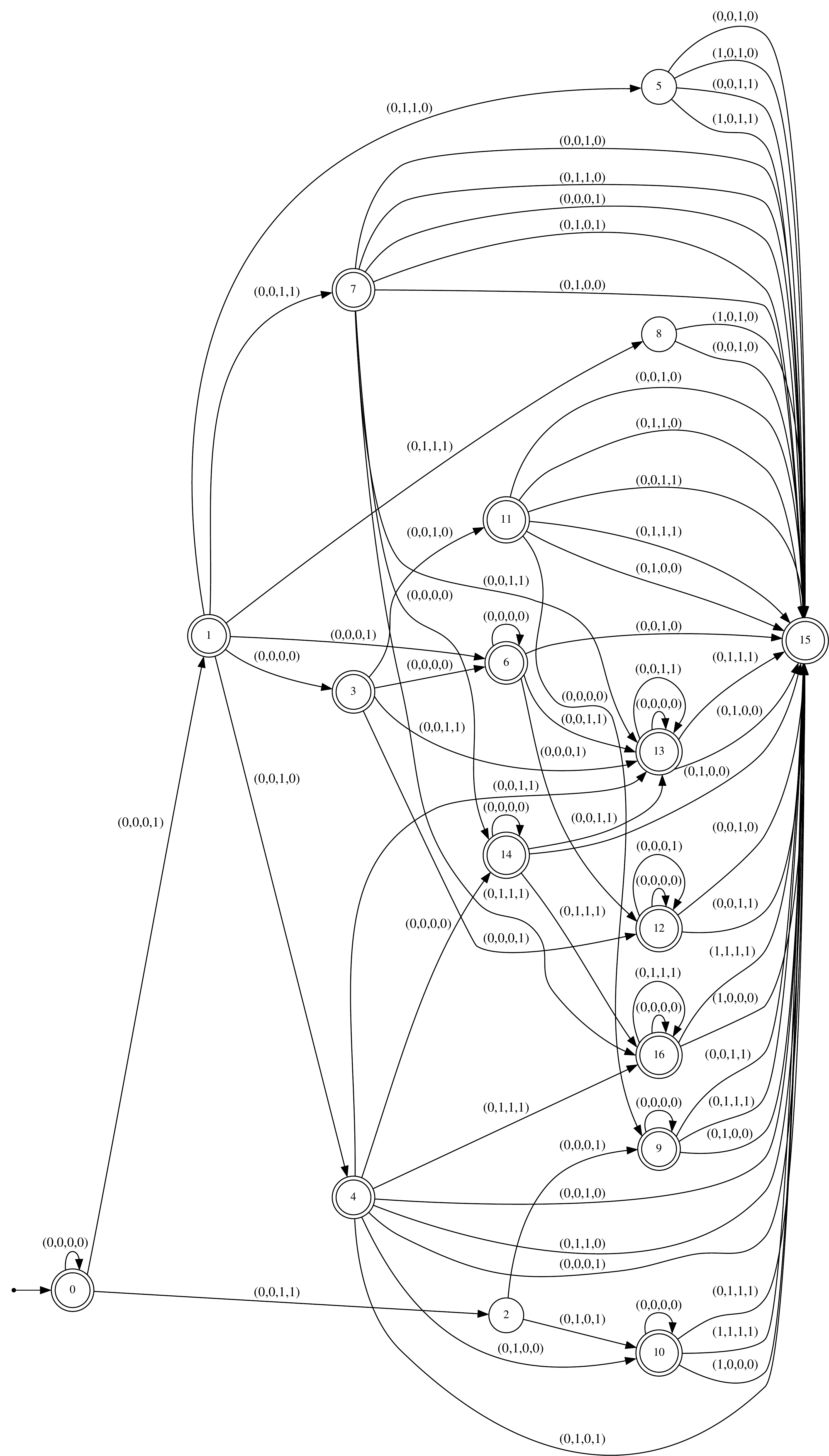}
  \caption{An automaton accepting the winning shift of the Rudin-Shapiro word.}\label{fig:rs}
\end{figure}

Let us then design an automaton accepting the encodings of all words in $W(X)$. We deviate slightly from
\autoref{sec:s_codable} and encode the occurrences of the letter $1$ in a word of $W(X)$ as a triple $(a, b, c)$ as
follows. We use $1$-based indexing and reserve the value $0$ to indicate that an occurrence is missing. For example,
$(0, b, c)$ with $b, c \neq 0$ means that a word has exactly two occurrences of $1$ in positions $b$ and $c$. We
require that $a \leq b \leq c$ and equality occurs only if both values are $0$. For example the words
$10001000000010^\omega$ and $001000000010^\omega$, both of which are in $W(X)$, are respectively encoded as
$(1, 5, 13)$ and $(0, 3, 11)$. The representations of the $3$-tuples in base $2$ are obtained by representing each
component in base $2$ and padding with $0$ to achieve common length. The following Walnut code builds an automaton
accepting the encodings $(a,b,c)$ of words of $W(X)$ as described.
\begin{verbatim}
def tm "(a = 0 & b = 0 & c = 0) |
        (a = 0 & b = 0 & c > 0 & En $isRS(c-1,n)) |
        (a = 0 & b > 0 & c > 0 & En $extRS2(b-1,c-1,n)) |
        (a > 0 & b > 0 & c > 0 & En1,n2 $extRS3(a-1,b-1,c-1,n1,n2))":
\end{verbatim}
The automaton produced by Walnut is depicted in \autoref{fig:tm}. We saw above that the coding dimension of $W(X)$ is
$3$. By a careful study of the automaton, we recover the following characterization of $W(X)$ described in
\cite[Sect.~2]{2019:on_winning_shifts_of_marked_uniform_substitutions}.
\begin{itemize}
  \item The coding dimension of $W(X)$ is $3$.
  \item If $\infw{x}$ in $W(X)$ contains three occurrences of $1$ at positions $a$, $b$, $c$ with $a < b < c$, then
        $a = 1$, $c - b = 2^k$ for some $k \geq 1$, and $b - 1 \leq 2^{k-1}$.
  \item If $\infw{x}$ in $W(X)$ contains exactly two occurrences of $1$ at positions $b$, $c$ with $b < c$, then either
        $b = 1$ or $c - b = 2^k$ for some $k \geq 1$ and $b - 1 \leq 2^{k-1}$.
  \item If $\infw{x}$ in $W(x)$ contains exactly one occurrence of $1$, then this occurrence can be at any position.
\end{itemize}

The above procedure can be repeated for other automatic words. We did this for the regular paperfolding word, the
Rudin-Shapiro word \cite[Sect.~5.1]{2003:automatic_sequences}, and the period-doubling word
\cite[Ex.~6.3.4]{2003:automatic_sequences} . The winning shift associated with the regular paperfolding word has coding
dimension $3$, and the automaton of \autoref{fig:paper} accepts it. The winning shift for the Rudin-Shapiro word has
coding dimension $4$ and automaton of \autoref{fig:rs}. Finally, the winning shift for the period-doubling word has
coding dimension $2$ and automaton of \autoref{fig:pd}. This automaton yields the following characterization for the
winning shift $W(X)$ of the period-doubling word:
\begin{itemize}
  \item The coding dimension of $W(X)$ is $2$.
  \item If $\infw{x}$ in $W(X)$ contains exactly two occurrences of $1$ at positions $a$, $b$ with $a < b$, then
        $b - a = 2^k$ for some $k \geq 1$ and $a - 1 \leq 2^{k-1}$.
  \item If $\infw{x}$ in $W(x)$ contains exactly one occurrence of $1$, then this occurrence can be at any position.
\end{itemize}

\section{Winning Shifts of Sofic Shifts and \texorpdfstring{$\omega$}{omega}-regular Languages}
Let us first provide a proof that $W(X)$ is regular whenever $X$ is. This follows from
\cite[Prop.~3.8]{2014:playing_with_subshifts} in the case of a binary alphabet. We provide a general proof for
completeness. It is a straightforward generalization of the arguments of \cite[Prop.~3.8]{2014:playing_with_subshifts}.
See \cite[Fig.~2]{2014:playing_with_subshifts} for example automata. For the basics on sofic shifts, we refer the
reader to \cite{1995:an_introduction_to_symbolic_dynamics_and_coding}. Observe that the sofic shifts considered here
are \emph{one-sided}.

\begin{proposition}\label{prop:WRegular}
  If $X$ is a regular language, then $W(X)$ is a regular language. In particular if $X$ is a sofic shift, then $W(X)$
  is a sofic shift.
\end{proposition}
\begin{proof}
  Denote by $\mathcal{A}$ a deterministic finite automaton $(Q, \Sigma, \delta, q_0, F)$ accepting a regular language
  $X$. Based on $\mathcal{A}$, we build a boolean automaton $\mathcal{B}$ (sometimes called an alternating automaton)
  accepting $W(X)$. Since boolean automata accept exactly the regular languages, we see that $W(X)$ is regular. For the
  proof of this fact and rigorous definition of boolean automata, see
  \cite{1980:on_equations_for_regular_languages_finite_automata} for example.

  Let $Q = \{q_1, \ldots, q_n\}$. The set of states of $\mathcal{B}$ is $Q$, and its set of final states is $F$. Let
  $\tau\colon Q \times \Sigma \to \mathbb{B}^Q$ be the transition function mapping a state and a letter to a boolean
  function. We define
  \begin{equation*}
    \tau(q, a) = \bigvee_{\substack{\{c_1, \ldots, c_{a+1}\} \subseteq \Sigma \\ |\{c_1, \ldots, c_{a+1}\}| = a+1}} \bigwedge_{i=1}^{a+1} \delta(q, c_i).
  \end{equation*}
  For example, if $\Sigma = \{0,1\}$ and $a = 1$, then $\tau(q, a) = s_0 \land s_1$ where $s_0 = \delta(q, 0)$ and
  $s_1 = \delta(q, 1)$. The interpretation is that the formula is true and the input $aw$ is accepted if and only if
  the computations on two copies of $\mathcal{A}$ with respective initial states $s_0$ and $s_1$ and input $w$ are both
  accepted. In other words, we substitute $s_0$ with $1$ if $\mathcal{A}$ accepts $w$ from $s_0$ and $0$ otherwise. We
  do the same for $s_1$ and then evaluate $s_0 \land s_1$.

  The function $\tau$ extends to $Q \times A^*$ recursively by setting $\tau(q, aw)$ to be the boolean function
  $f_{q,a}(\tau(q_1, w), \ldots, \tau(q_n, w))$ where $f_{q,a} = \tau(q, a)$ and $\tau(q, \varepsilon) = q$.
  We accept a word $w$ from state $q$ if and only if $\tau(q, w)$ evaluates to $1$.

  Let $L_q^{\mathcal{X}}$ be the language an automaton $\mathcal{X}$ accepts from state $q$. First of all, we have
  $\varepsilon \in W(L_q^\mathcal{A})$ if and only if $\varepsilon \in L_q^\mathcal{B}$ because the final states in
  $\mathcal{A}$ and $\mathcal{B}$ are the same. Let $a$ be a letter and $w$ a word. Then
  \begin{align*}
    aw \in W(L_q^\mathcal{A})
    &\quad \Longleftrightarrow \quad
    \exists C \subseteq \Sigma\colon |C| = a+1 \land \forall c \in C\colon w \in W\left(L_{\delta(q,c)}^\mathcal{A}\right) \\
    &\quad \Longleftrightarrow \quad
    \exists C \subseteq \Sigma\colon |C| = a+1 \land \forall c \in C\colon w \in L_{\delta(q,c)}^\mathcal{B} \\
    &\quad \Longleftrightarrow \quad
    \bigvee_{\substack{\{c_1, \ldots, c_{a+1}\} \subseteq \Sigma \\ |\{c_1, \ldots, c_{a+1}\}| = a+1}} \bigwedge_{i=1}^{a+1} \tau(\delta(q,c_i), w) = 1 \\
    &\quad \Longleftrightarrow \quad
    \tau(q, aw) = 1 \\
    &\quad \Longleftrightarrow \quad
    aw \in L_q^\mathcal{B}
  \end{align*}
  where the second equivalence follows from the induction hypothesis. The first claim follows.
  
  The latter claim follows since $W(X)$ is shift-invariant and closed whenever $X$ is. (This is easy to show; see
  \cite[Section~3]{2014:playing_with_subshifts} for details.)
\end{proof}

If the automaton $\mathcal{A}$ in the preceding proof is a B{\"u}chi automaton accepting a $\omega$-language $X$, then
the constructed automaton $\mathcal{B}$ is a boolean B{\"u}chi automaton accepting $W(X)$. Since there is an equivalent
B{\"u}chi automaton for each boolean B{\"u}chi automaton \cite{1984:alternating_finite_automata_on_omega-words}, we
obtain the following result.

\begin{proposition}
\label{prop:WOmegaRegular}
  Let $X$ be a subset of $A^\N$. If $X$ is $\omega$-regular, then $W(X)$ is $\omega$-regular.
\end{proposition}

The natural definition of $W(X)$ in this case is through branching structures of trees
\cite{2021:trees_in_positive_entropy_subshifts}, and this is exactly analogous to how acceptance is defined in boolean
$\omega$-automata, so the proof is exactly analogous. There is a small technical detail: Unlike in
\cite{1980:on_equations_for_regular_languages_finite_automata}, the definition of a boolean automaton in
\cite{1984:alternating_finite_automata_on_omega-words} does not allow an arbitrary boolean combination but only a
single conjunction or disjunction at each state (and $\varepsilon$-transitions are not allowed either). However, since
$\omega$-regular languages are closed under inverse substitutions, it is enough to show that the image of $W(X)$ under
the substitution $s \mapsto s\#$ is $\omega$-regular, allowing the simulation of an expression of the form
$\bigvee \bigwedge$. That is, we use two transitions per one letter and switch back and forth between an automaton with
conjunctions and an automaton with disjunctions.

Notice that when $W(X)$ is $\ans{S}$-codable, it is not true that $W(X)$ is necessarily $\omega$-regular. Consider for
example the winning shift $W(X)$ of the Thue-Morse word studied in \autoref{sec:automata}. Suppose for a
contradiction that there exists a B{\"u}chi automaton $\mathcal{A}$ recognizing $W(X)$. We first observe that if in a
successful run a final state is visited twice, then the path corresponding to this cycle has label in $0^*$. Otherwise
by pumping we could produce infinitely many occurrences of $1$ to a word in $W(X)$ and this is impossible as the coding
dimension of $W(X)$ is $3$. Since we may assume that each final state is visited at least twice in some successful run,
it follows that $\mathcal{A}$, when consider as a finite automaton, accepts $L 0^*$ where
\begin{equation*}
  L = \{w \in \{0,1\}^* 1 : w0^\omega \in W(X)\}.
\end{equation*}
Hence $L$ is regular. Let $L'$ be the regular language of words belonging to $L$ and having exactly three occurrences
of $1$. By the characterization provided in \autoref{sec:automata}, we have
\begin{equation*}
  L' = \{10^n 10^m1 : \text{$m - 1 = 2^k$ for some $k \geq 1$ and $n \leq 2^{k-1}$}\}.
\end{equation*}
A standard argument using the pumping lemma shows that $L'$ is not regular; a contradiction.

\begin{definition}\label{def:codable_general}
  Let $\ans{S}$ be an ANS and $Y$ be a subset of $A^\N$. We say that $Y$ is (weakly) $\ans{S}$-codable if there exists
  a bijection $\pi\colon A \to \{0, 1, \ldots, |A| - 1\}$ such that $0 \in \pi(A)$ and $\pi(Y)$ is (weakly)
  $\ans{S}$-codable. We call $\pi^{-1}(0)$ the \emph{zero symbol} of $Y$ and $\pi^{-1}(0^\omega)$ the \emph{zero point}
  associated with $Y$. All other symbols are \emph{nonzero symbols}.
\end{definition}

\begin{definition}
  We say that a set is (weakly) \emph{$1$-codable} if it is (weakly) $\ans{S}$-codable for the ANS $\ans{S}$ with
  language $0^*$ and radix order. When we say that a set is $1$-recognizable, we simply mean that it is
  $\ans{S}$-recognizable.
\end{definition}

We show in \autoref{thm:CountableSoficFCD} that for any sofic shift $X$, the winning shift $W(X)$ is weakly
$1$-codable, and in the countable case it is $1$-codable.

\begin{lemma}\label{lem:UniquePeriodicPoint}
  If $X$ is a sofic shift with a unique periodic point, then $X$ is countable and has finite coding dimension.
\end{lemma}
\begin{proof}
  Say $X$ is a sofic shift. We may assume that $X \subseteq \N^\N$. If $X$ has a periodic point with period greater
  than $1$, then it does not have unique periodic point. We may thus assume that the unique periodic point is
  $0^\omega$, and we take $0$ as the zero symbol. Suppose for a contradiction that there are words in $X$ with
  arbitrarily large sums. Recall that by the definition of soficity, $X$ is the set of edge labels of right-infinite
  paths in some finite graph. Fix such a graph defining $X$, and let $M$ be the number of vertices in this graph.
  
  Let $w_n$ be a word in the language of $X$ with sum at least $n$. In particular, because the alphabet of $X$ is
  finite, the support of $w_n$ tends to infinity in $n$. It is easy to see that if the support of $w_n$ is strictly
  larger than $M$, then the path corresponding to $w_n$ contains a cycle with a nonzero label (apply the pigeonhole
  principle to positions with a nonzero label).  This cycle can be used to construct a periodic point which is not
  equal to $0^\omega$, and thus there are at least two periodic points, a contradiction.
  
  If $X$ has finite coding dimension, then $X$ is clearly countable.
\end{proof}

In the proof, one can also use Ogden's lemma \cite{1968:a_helpful_result_for_proving_inherent_ambiguity} from formal
language theory.

We recall the connection between sofic shifts and $\omega$-regular languages as this is needed in the following
sections. (This would allow one to deduce the latter claim of \autoref{prop:WRegular} from
\autoref{prop:WOmegaRegular}.)

\begin{lemma}\label{lem:OmegaRegularSofic}
  A subshift is sofic if and only if it is $\omega$-regular.
\end{lemma}
\begin{proof}
  A subshift $X$ is sofic if and only if it is the set of labels of right-infinite paths in a finite graph. If $X$ is
  sofic, we can directly use this graph as a deterministic B{\"u}chi automaton for $X$ as an $\omega$-language.
  Conversely, suppose $X$ is $\omega$-regular, so that it is accepted by a B{\"u}chi automaton. Without loss of
  generality, we may assume that from every state of the automaton there is a path to a final state. Then it is clear
  that making all the states final does not change the language (since $X$ is topologically closed) and making all
  states initial does not change the language either (since $X$ is shift-invariant). What is left is nothing but a
  finite edge-labeled graph, and the labels of its right-infinite paths form precisely $X$.
\end{proof}

\begin{proposition}\label{prop:UPP1Codable}
  A subshift is sofic with a unique periodic point if and only if it is $1$-codable.
\end{proposition}
\begin{proof}
  Let $X$ be a subshift. We may assume that $X \subseteq \N^\N$ and $X$ has zero symbol $0$. Say $X$ is
  $1$-codable. Since $X$ has finite coding dimension, it must contain a unique periodic point. By $1$-codability, the
  sets $P_v(X)$ are $1$-recognizable for all $v \in \N_{>0}^*$. Let $v$ in $\N_{>0}^*$ be fixed and write
  $v = v_0 \dotsm v_{\ell-1}$ with $\ell = |v|$. We have a finite automaton accepting the tuples
  \begin{equation*}
    (\#^{\ell-1-n_0} 1^{n_0}, \#^{\ell-1-n_1} 1^{n_1}, \ldots, 1^{\ell-1})
  \end{equation*}
  such that $0^{n_0} v_0 0^{n_1-n_0-1} v_1 \dotsm v_{\ell-1} 0^\omega \in Q_v(X)$ (recall that
  $Q_v(X) = e^{-1}(v) \cap X$). Because regular languages are closed under reversal, we also have an automaton
  $\mathcal{A}$ accepting the reversals
  \begin{equation*}
    (1^{n_0} \#^{\ell-1-n_0}, 1^{n_1} \#^{\ell-1-n_1}, \ldots, 1^{\ell-1})
  \end{equation*}
  of such word tuples. Clearly there is a bijective transducer that reduces $\omega$-words of the form
  $0^{n_0} v_0 0^{n_1-n_0-1} v_1 \dotsm v_{\ell-1} 0^\omega$ to words of the form
  $(1^{n_0} \#^\omega, 1^{n_1} \#^\omega, \ldots, 1^{\ell-1} \#^\omega)$, and the automaton $\mathcal{A}$ yields an
  automaton accepting these $\omega$-words. It follows that $Q_v(X)$ is $\omega$-regular. Then $X$ is also
  $\omega$-regular as a finite union of such $Q_v(X)$. Since $X$ is a subshift, it is sofic by
  \autoref{lem:OmegaRegularSofic}.

  The other direction is exactly analogous. If $X$ is sofic, it is $\omega$-regular, and from its automaton we obtain
  another automaton accepting finite words $w = 0^{n_0} v_0 0^{n_1-n_0-1} v_1 \dotsm v_{\ell-1}$ such that
  $w 0^\omega \in Q_v(X)$ by additionally keeping track of the word formed by the nonzero symbols. After a
  transduction, we can accept the codings $(1^{n_0} \#^\omega, 1^{n_1} \#^\omega, \ldots, 1^{\ell-1} \#^\omega)$ of
  such words, and after a reversal we see that $X$ is weakly $1$-codable. It follows from
  \autoref{lem:UniquePeriodicPoint} that $X$ has finite coding dimension. Therefore $X$ is $1$-codable.
\end{proof}

Of course, the proof really shows that a set of $1$-codable if and only if it is $\omega$-regular and has finite coding
dimension, but our interest in this paper is in closed sets.

\begin{theorem}\label{thm:CountableSoficFCD}
  If $X$ is a countable sofic shift, then $W(X)$ is $1$-codable. More generally, $W(X)$ is weakly $1$-codable for every
  sofic shift $X$.
\end{theorem}
\begin{proof}
  We may again assume that $X \subseteq \N^\N$ and $X$ has zero symbol $0$. A sofic shift is countable if and only if
  it has zero entropy. Thus, when $X$ is a countable sofic shift, $W(X)$ is a hereditary sofic shift with zero entropy
  \cite[Prop.~5.7]{2014:playing_with_subshifts}. Say $W(X)$ has a periodic point other than $0^\omega$. Suppose that
  the word corresponding to the minimum period $p$ of this periodic point contains $\ell$ occurrences of nonzero
  letters. By downgrading each of these nonzero letters independently, we find that there are at least $2^{k\ell}$
  factors of length $kp$ in words of $X$ for all $k$. It follows that $W(X)$ has positive entropy; a contradiction.
  Thus the only periodic point of $W(X)$ is $0^\omega$. It follows from \autoref{lem:UniquePeriodicPoint} that $X$ has
  finite coding dimension.

  For a general sofic shift $X$, the winning shift $W(X)$ is sofic by \autoref{prop:WRegular}. To see that $W(X)$ is
  weakly $1$-codable, it is enough to show that its restriction to points with sum at most $d$ is $1$-codable for all
  $d$. This restriction is clearly a sofic shift with a unique periodic point, and \autoref{prop:UPP1Codable} applies.
\end{proof}

By the results of \cite{2012:multidimensional_sets_recognizable_in_all_abstract_numeration}, this means that the
winning shift $W(X)$ of a countable sofic shift $X$ is $\ans{S}$-codable for every regular ANS $\ans{S}$ with radix
order.

\section{Robustness}\label{sec:Robust}
\autoref{thm:main} is only interesting if one agrees that the class of $\ans{S}$-codable subshifts is a natural one. In
this section, we prove some closure properties for the class of (weakly) $\ans{S}$-codable sets for addable $\ans{S}$,
which imply that this class is relatively robust. First, while $\ans{S}$-codability is defined combinatorially, it is a
property of the abstract dynamical system in the sense that it is preserved under conjugacy.

\begin{proposition}\label{prop:Conjugacy}
  Let $\ans{S}$ be an addable ANS, and consider a (weakly) $\ans{S}$-codable subshift $X$ such that
  $X \subseteq \N^\N$. If $Y$ is conjugate to $X$ and the conjugating map maps the zero point associated with $Y$ to
  the zero point $0^\omega$, then $Y$ is (weakly) $\ans{S}$-codable.
\end{proposition}

Our main examples of $\ans{S}$-codable subshifts are winning shifts, which are always hereditary, so the following
seems natural to prove.

\begin{proposition}\label{prop:HereditaryClosure}
  Let $\ans{S}$ be an addable ANS and $Y$ be a subset of $\N^\N$. If $Y$ is $\ans{S}$-codable, then the hereditary
  closure of $Y$ is $\ans{S}$-codable.
\end{proposition}

Finally, our class contains a subclass of sofic shifts generalizing one direction of \autoref{prop:UPP1Codable}.

\begin{proposition}\label{prop:CountableSofic}
  Let $\ans{S}$ be an addable ANS. Every sofic shift with a unique periodic point is $\ans{S}$-codable. More generally,
  every sofic shift is weakly $\ans{S}$-codable with respect to any zero symbol.
\end{proposition}

Note that the latter claim is not as weak as it may seem: if $(Y, \sigma)$ is sofic, so is $(Y, \sigma^n)$ for any
$n \geq 1$ (here $\sigma$ is the shift map), and thus these are also weakly $\ans{S}$-codable. Since eventually
periodic points are dense in a sofic shift, in a sense the weak codability of these subshifts codes all of its points.
To have a more useful statement, one might want to strengthen weak codability with uniformity conditions, but this is
beyond the scope of the present paper.

We deduce the above three propositions from a more abstract \autoref{prp:abstract}.

Let us use infix notation for relations. For sets $X \subseteq A^\N, Y \subseteq B^\N$ and a relation
$R \subseteq A^\N \times B^\N$, define
\begin{equation*}
  xR = \{y \in B^\N : (x,y) \in R\} \subseteq B^\N \quad \text{and} \quad XR = \bigcup_{x \in X} xR.
\end{equation*}
Symmetrically we define $Ry \subseteq A^\N$ for $y \in Y$ and $RY$. To such $R$ we also
associate a relation $R^\Sigma$ in $\N \times \N$ by setting
\begin{equation*}
  m R^\Sigma n \iff \exists (x,y) \in R\colon \sum x = m \wedge \sum y = n.
\end{equation*}
A relation $R' \subseteq \N \times \N$ is \emph{locally finite} if $|mR'| < \infty$ and $|R'n| < \infty$ for all
$m, n \in \N$. We say $R \subseteq A^\N \times B^\N$ is locally finite if $R^\Sigma$ is locally finite.
 
\begin{proposition}\label{prp:abstract}
  Let $X \subseteq A^\N, Y \subseteq B^\N$ be subshifts and $R \subseteq A^\N \times B^\N$ a sofic shift such that
  $R \cap X \times Y$ is locally finite. If $Y = XR$ and $X$ is weakly $\ans{S}$-codable, then $Y$ is weakly
  $\ans{S}$-codable.
\end{proposition}
\begin{proof}
  We may assume that $A, B \subseteq \N$ and $0$ is the zero symbol in both $X$ and $Y$. It suffices to show that
  $P_v(Y)$ is $\ans{S}$-recognizable for each $v \in \N_{>0}^*$. Fix $v \in \N_{>0}^*$. By the assumption of local
  finiteness, $Q_v(Y)$ is a union of finitely many sets of the form $Q_u(X) R \cap Q_v(B^\N)$ where $u \in A^*$, so it
  suffices to show that each $Q_u(X) R \cap Q_v(B^\N)$ is $\ans{S}$-codable. Since
  $P_v(Q_u(X) R \cap Q_v(B^\N)) = P_v(Q_u(X) R)$, it suffices to prove that $P_v(Q_u(X) R)$ is $\ans{S}$-recognizable
  for all $u \in A^*$.

  Suppose $|u| = k, |v| = \ell$, and define
  \begin{align*}
    T = \{(m_1, \ldots, m_k, n_1, \ldots, n_\ell) : \exists \infw{x} \in X, \infw{y} \in Y\colon e(\infw{x}) &= u, e(\infw{y}) = v, \\
                                                                                                 s(\infw{x}) &= (m_1,\ldots,m_k), s(\infw{y}) = (n_1,\ldots,n_\ell) \}.
  \end{align*}
  Clearly it suffices to show that $T$ is $\ans{S}$-recognizable since $P_v(Q_u(X) R)$ is just the projection of $T$ to
  the last $\ell$ coordinates.

  To accept $\rep(T)$, we fix an ordering for the components. Given an element $(b_1, \ldots, b_{k+\ell})$ of $T$,
  there exists a permutation $P$ that transforms this element to $(c_1, \ldots, c_{k+\ell})$ that $c_i = b_{P(i)}$ for
  all $i$ and $c_1 \leq \ldots \leq c_{k+\ell}$. Moreover, if a component $b_i$ of $(b_1, \ldots b_k)$ equals a
  component $b_j$ of $(b_{k+1}, \ldots, b_{k+\ell})$, we may require that $P^{-1}(i) < P^{-1}(j)$. This makes $P$
  unique since by definition $m_1 < \ldots < m_k$ and $n_1 < \ldots < n_\ell$. Let $T_P$ the set of elements of $T$
  with a common permutation $P$. The set $T_P$ is $\ans{S}$-recognizable if and only if $P(T_P)$ is
  $\ans{S}$-recognizable as a simple transducer can reorder components when $P$ is fixed. As $T$ is a union of the sets
  $T_P$ and there are only finitely many permutations $P$, it suffices to prove that $P(T_P)$ is $\ans{S}$-recognizable
  for a fixed $P$.

  Define $L = \rep(P(T_P))$ so that $L$ has elements $(w_1, \ldots, w_{k+\ell})^\#$ with $\val(w_i) = c_i$ for all $i$.
  Let us perform a transduction to transform $(w_1, \ldots, w_{k+\ell})^\#$ to $(w_1, t_2, \ldots, t_{k+\ell})^\#$
  where $\val(w_i) + \val(t_{i+1}) = \val(w_{i+1})$ for all $i$. Such a transduction is possible because $\ans{S}$ is
  addable. Since regular languages are closed under inverse transductions, the language $L$ is regular if and only if
  its transduction $L'$ is regular.

  The tuple $(w_1, t_2, ..., t_{k + \ell})^\#$ in $L'$ represents the support of an infinite word over the alphabet
  $A \times B$, and the first nonzero symbol appears at position $\val(w_1)$, the second at $\val(w_1)+\val(t_2)$, the
  third at $\val(w_1)+\val(t_2)+\val(t_3)$, and so on. The permutation $P$ tells us which of these positions have a
  symbol on the first or second component of the alphabet $A \times B$. An important interpretation detail is that
  elements of $(A \setminus \{0\}) \times (B \setminus \{0\})$ are represented by having $t_{i+1} = 0$ and having
  exactly one of the values $P(i), P(i+1)$ be in $\{1, \ldots, k\}$. This follows from the definition of $P$. In a
  sense, the permutation $P$ represents a shuffle of the words $u$ and $v$.

  We now proceed to run the finite-state automaton defining the sofic shift $R$ on the encoded input. From
  \autoref{lem:OmegaRegularSofic}, we see that sofic shifts are $\omega$-regular, so let $\mathcal{A}$ be a
  deterministic automaton for $R$ with any standard acceptance condition. Observe that a run of $\mathcal{A}$ on an
  infinite word with support of size at most $k+\ell$ is determined entirely by its states corresponding to at most
  $k+\ell$ nonzero letters.

  Given input $\rep(n)$, we can compute the state $q'$ the automaton $\mathcal{A}$ is in after reading input $0^n$ from
  any state $q$. If $n$ is small enough that $\mathcal{A}$ does not enter a loop, then we can use a lookup table.
  Suppose $n$ is not small. Let $K$ be the least common multiple of the lengths of the loops $\mathcal{A}$ may enter.
  Simulate $\mathcal{A}$ for $k$ steps until a loop is reached, compute in parallel $n - k$ modulo $K$ (possible since
  $\ans{S}$ is addable), and use again a lookup table.

  We describe an NFA $\mathcal{B}$ as follows. On each component of the input $(w_1, t_2, \ldots, t_{k+\ell})^\#$, the
  automaton $\mathcal{B}$ simulates a computation of $\mathcal{A}$ in parallel. On the $i$th component, $i \geq 2$, we
  guess a state $q'_i$ of $\mathcal{A}$ and find a state $q$ such that $\mathcal{A}$ is at state $q$ after reading
  input $0^{\val(t_i)}$ like in the previous paragraph. Let $x$ be the letter over $A \times B$ such that the infinite
  word encoded by the input has letter $x$ at position $\val(w_1) + \val(t_1) + \dotsm + \val(t_i)$. Based on $P$, $u$,
  and $v$, we can check if there is a transition to state $q_i$ with letter $x$ (if not, we reject the computation). On
  the first component, we do the same, but we read $0^{\val(w_1)}$ and start from the initial state of $\mathcal{A}$.
  When the computations have finished, we can verify if $q'_i = q_{i-1}$ for $i = 2, \ldots, k+\ell$. If this is the
  case, there is a path from the initial state of $\mathcal{A}$ to state $q_{k+\ell}$ whose label corresponds to the
  prefix of length $k + \ell + \val(w_1) + \val(t_1) + \dotsm + \val(t_{k+\ell})$ of the encoded infinite word.
  Finally, we check if $q_{k+\ell}$ belongs to the finite set of states of $\mathcal{A}$ which accept when only zeros
  are read and accept if and only this is so. (Note that in all the standard acceptance conditions, acceptance is a
  shift-invariant tail event, even if the set of valid runs need not be shift-invariant in general.)

  By construction, the automaton $\mathcal{B}$ accepts, with respect to $P$, $u$, and $v$, the encodings of infinite
  words in $R$ such that the support of the first component spells the word $u$ and the support of the second component
  the word $v$. In other words, the automaton accepts the language $L'$. The claim follows.
\end{proof}

We can now prove our claims.

\begin{proof}[Proof of \autoref{prop:Conjugacy}]
  Consider first weak $\ans{S}$-codability. Suppose that $\phi\colon Y \to X$ is a conjugacy and $\phi$ maps the zero
  point associated with $Y$ to $0^\omega$. We select $R$ to be the graph of the conjugacy. By \autoref{prp:abstract},
  it suffices to verify that $R$ is locally finite. 

  Let $r$ be a biradius of $\phi$ (a radius common to both $\phi$ and $\phi^{-1}$). Since the zero point associated
  with $Y$ maps to $0^\omega$, the mapping $\phi$ can produce at most $2r + 1$ new nonzero letters for each nonzero
  letter in a word of $Y$. Therefore if $\infw{y}$ in $Y$ contains $k$ nonzero letters, then $\phi(\infw{y})$ contains
  at most $(2r+1)k$ nonzero letters. Therefore $|m R^\Sigma| \leq (2r + 1)m$ for all $m$. An analogous argument shows
  that $|R^\Sigma m| \leq (2r+1)m$ for all $m$, so it follows that $R^\Sigma$ is locally finite.
  
  For the complete claim, simply observe, again, that if $\phi\colon X \to Y$ is a conjugacy, $X$ has bounded sums if
  and only if $Y$ does.
\end{proof}

\begin{proof}[Proof of \autoref{prop:CountableSofic}]  
  Let $Y$ be a sofic shift with a unique periodic point. It is clear that this periodic point must be $a^\omega$ for a
  letter $a$ for otherwise we would have more than one periodic point. Select a permutation $\pi$ such that
  $\pi(a) = 0$. By \autoref{lem:UniquePeriodicPoint}, $\pi(Y)$ has finite coding dimension. For the first claim, it now
  suffices to to show that $\pi(Y)$ is weakly $\ans{S}$-codable. Select $R = \{0^\omega\} \times \pi(Y)$. Since
  $\pi(Y)$ has finite coding dimension, it is plain that $|0 R^\Sigma|$ and $|R^\Sigma 0|$ are finite, so $R$ is
  locally finite. The claim follows from \autoref{prp:abstract} as $\{0^\omega\}$ is clearly a weakly $\ans{S}$-codable
  sofic shift.
  
  For the general case, say $Y$ is sofic and $a^\omega \in Y$ for a letter $a$. Select a permutation $\pi$ such that
  $\pi(a) = 0$. The restriction $Z$ of $\pi(Y)$ to words with sum at most $d$ is also sofic for all $d$. By the above,
  $Z$ is $\ans{S}$-codable, and it follows that $Y$ is weakly $\ans{S}$-codable.
\end{proof}

%Observe that \autoref{prop:CountableSofic} is in a sense a best possible result. If a sofic shift $X$ contains two
%periodic points $\infw{x}$ and $\infw{y}$, then for all permutations $\pi$, as in \autoref{def:codable_general}, either
%$\sum \pi(\infw{x}) = \infty$ or $\sum \pi(\infw{y}) = \infty$. Thus $X$ is not $\ans{S}$-codable for any ANS
%$\ans{S}$. %The subshift $X$ could be weakly $\ans{S}$-codable, but this is not so interesting in the presence of words with infinite sums.

\begin{proof}[Proof of \autoref{prop:HereditaryClosure}]
  Suppose that $Y$ is $\ans{S}$-codable, and let $\tilde{Y}$ be its hereditary closure. Define
  \begin{equation*}
    R = \{(\infw{x}, \infw{y}) \in \N^\N \times \N^\N : \text{$\infw{x} = x_0 x_1 \dotsm$, $\infw{y} = y_0 y_1 \dotsm$, and $x_i \geq y_i$ for all $i$}\}.
  \end{equation*}
  The set $R$ is clearly a sofic shift. Moreover, the relation $R \cap (Y \times \tilde{Y})$ is locally finite since
  the coding dimension of $Y$ is finite. The claim follows from \autoref{prp:abstract}.
\end{proof}

\autoref{prop:HereditaryClosure} is not true under weak $\ans{S}$-codability. One can obtain examples by assuming that
$Y$ has no words with finite support. It is possible to imagine slightly modified definitions that disallow this, so we
give a more interesting example.

\begin{example}
  Let $N$ be a subset of $2\N$, and let $X_N$ be the smallest subshift containing the infinite words
  \begin{equation*}
    0^m 10^{2n} \prod_{i = 1}^n 10^{2n-2i+1} \cdot 0^\omega
  \end{equation*}
  where $m \in \N$, $2n \in N$. Clearly the hereditary closure $\tilde{X}_N$ of $X_N$ contains the infinite word
  $10^{2n}10^{\omega}$ if and only if $2n \in N$. Let
  \begin{equation*}
    Z = \{\infw{x} \in \tilde{X}_N : \sum \infw{x} = 2\}.
  \end{equation*}
  Since there are uncountably many choices for $N$ and there are only countably many $\ans{S}$-codable subsets of
  $\{0, 1\}^\N$, we can choose $N$ so that $Z$ is not $\ans{S}$-codable. Since $Z$ is $\ans{S}$-codable if
  $\tilde{X}_N$ is weakly $\ans{S}$-codable, we find the desired counterexample if we show that $X_N$ is always weakly
  $\ans{S}$-codable.

  Let $Y_n$ be the set $\{\infw{x} \in X_n : \sum \infw{x} \leq n\}$. It suffices to show that $Y_n$ is
  $\ans{S}$-codable for all $n$. It is straightforward to see that there exist finitely many words $w$ such that the
  words of $Y_n$ are of the form $0^m w 0^\omega$ for $m \in \N$. The closure of the union of words of such form is
  clearly a sofic shift with exactly one periodic point. \autoref{prop:CountableSofic} implies that $Y_n$ is
  $\ans{S}$-codable when $\ans{S}$ is addable.
\end{example}

\begin{remark}\label{rem:Addability}
  We have used the assumption of addability of the ANS $\ans{S} = (L, \prec)$ in this section, since this is the
  natural generality for computing winning shifts. However, the proof of \autoref{prp:abstract} (and thus all the
  results in this section) really only need the property that given two words $u$ and $v$ in $L$, the difference
  $\val(u) - \val(v)$ can be computed modulo a constant by a finite-state machine, or equivalently the value $\val(u)$
  can be computed modulo a constant for a single word. It is known that this is true in the classical case where
  $\prec$ is the radix order \cite{2001:numeration_systems_on_a_regular_language}. We believe that this generalizes to
  every comparable ANS, but this is beyond the scope of this paper. Applied to the ANS on $0^*$ with radix order, this
  could be used to give another proof for the implication ``sofic with a unique periodic point $\implies$ $1$-codable''
  of \autoref{prop:UPP1Codable}.
\end{remark}

\section{Open problems}

\begin{question}
  Is there an addable ANS $\ans{S}$ such that some $\ans{S}$-automatic word has a winning shift with unbounded sums?
\end{question}

If the answer is negative, then one can drop the assumption of sublinear complexity in \autoref{thm:main} (the only
point of this property is that it is the most natural word-combinatorial assumption we know that implies bounded sums).
Cassaigne's example provides only a comparable ANS with this property.

For a general substitutive subshift we do not know how to describe even the bounded sum restrictions of winning shifts
in general. It seems natural to guess that the winning shift is somehow ``substitutive'' in this case too, and we
propose the following question:

\begin{problem}
  Devise an effective procedure that, given a fixed point $\infw{x}$ of a substitution $\tau$ and a word $v$, computes
  a finite description of $P_v(W(X))$ where $X$ is the subshift generated by $\infw{x}$.
\end{problem}

An obvious candidate for a finite description would be to show recognizability with respect to the corresponding
Dumont-Thomas numeration system for $P_v(W(X))$, but we do not see why this would hold.

When the winning shift of a subshift has infinite coding dimension, we do not obtain finite description of the winning
shift even in the addable case, only weak $\ans{S}$-codability.

\begin{problem}
  Find a finitary way to describe the winning shift in the weakly $\ans{S}$-recognizable case. In particular, describe
  $W(Z)$ for the subshift $Z$ from \autoref{prp:z_icd}.
\end{problem}

Besides these theoretical problems, it would be of interest to try to extend the practical computations in
\autoref{sec:Concrete} to examples where the winning shift has larger coding dimension. We expect that the methods
scale very badly, but this intuition has often turned out to be wrong in the setting of automatic theorem-proving; see
\cite[Remark~3]{2013:automatic_theorem-proving_in_combinatorics_on_words} and
\cite{2016:decision_algorithms_for_fibonacci-automatic_words_i} and its references.

We also mention two somewhat tangential problems arising from the considerations in \autoref{sec:Robust}. First, if $X$
and $Y$ are conjugate subshifts, it is not clear to us whether $W(X)$ and $W(Y)$ are somehow related to each other as
well (admittedly, it is not clear what a negative answer could be). Second, we conjecture that the results of
\autoref{sec:Robust} extend to arbitrary comparable ANS $\ans{S}$.

\section*{Acknowledgements}
The first author is grateful for his postdoc position in TCSMT which allowed him to freely focus on research in
2018--2021. The second author was supported by the Academy of Finland grant 2608073211.

\printbibliography

\end{document}